\documentclass[letterpaper, one column, 12pt]{IEEEtran}

\usepackage{graphics}
\usepackage{color}
\usepackage{amsmath}
\usepackage{amssymb}
\usepackage{graphicx}
\usepackage{cite}
\usepackage{epstopdf,datetime}
\newtheorem{theorem}{\bf Theorem}
\newtheorem{lemma}{\bf Lemma}

\newtheorem{definition}{\bf Definition}
\newtheorem{corollary}{\bf Corollary}

\DeclareMathAlphabet{\mathssf}{OT1}{cmss}{m}{sl}

\newcommand{\avgcost}{\overline{\mathcal{J}}}
\newcommand{\cost}{\mathcal{J}}
\newcommand{\m}[1]{\mathbf{#1}^m}
\newcommand{\instvec}{\m{X}_0}
\newcommand{\xonevec}{\m{X}_1}
\newcommand{\lo}[1]{\log_2\left(#1\right)}

\newcommand{\expectp}[2]{\mathbb{E}_{#1}\left[{#2}\right]}

\newcommand{\corr}{{\sigma_{XV}}}
\newcommand{\expect}[1]{\mathbb{E}\left[#1\right]}
\newcommand{\whatmn}[1]{\mathbf{\widehat{#1}}^m}
\newcommand{\wtildemn}[1]{\mathbf{\widetilde{#1}}^m}

\title{Information embedding and the triple role of control}
\author{Pulkit $\text{Grover}^\nabla$, Aaron B. $\text{Wagner}^\ddagger$ and Anant $\text{Sahai}^\dagger $\thanks{$\nabla$ Electrical and Computer Engineering, Carnegie Mellon University. Email: pulkit\;@\;cmu.edu. $\dagger${Electrical Engineering and Computer Sciences, University of California, Berkeley. Email: sahai\;@\;eecs.berkeley.edu}. $\ddagger$ Electrical and Computer Engineering, Cornell University. Email: wagner\;@\;ece.cornell.edu. An abridged version of this paper was  presented at the 2010 Information Theory Workshop (ITW), Cairo.    }}
\begin{document}\maketitle
%\today\;\;\currenttime
\begin{abstract}
We consider the problem of information embedding where the encoder modifies a white Gaussian host signal in a power-constrained manner to encode a message, and the decoder recovers both the embedded message and the \textit{modified} host signal. This partially extends the recent work of Sumszyk and Steinberg to the continuous-alphabet Gaussian setting.  Through a control-theoretic lens, we observe that the problem is a minimalist example of what is called the ``triple role'' of control actions. We show that a dirty-paper-coding strategy achieves the optimal rate for perfect recovery of the modified host and the message for any message rate. For imperfect recovery of the modified host, by deriving bounds on the minimum mean-square error (MMSE) in recovering the modified host signal, we show that DPC-based strategies % check if you want to qualify this.
 are guaranteed to attain within a uniform constant factor of $16$ of the optimal weighted sum of power required in host signal modification and the MMSE in the modified host signal reconstruction for all weights and all message rates. When specialized to the zero-rate case, our results provide the tightest known lower bounds on the asymptotic costs for the vector version of a famous open problem in decentralized control: the Witsenhausen counterexample. Numerically, this tighter bound helps us characterize the asymptotically optimal costs for the vector Witsenhausen problem to within a factor of $1.3$ for all problem parameters, improving on the earlier best known bound of $2$.
\end{abstract}

\section{Introduction}

\label{sec:intro}

The problem of interest in this paper (see Fig.~\ref{fig:blockdgm}) derives its motivation from both information theory and decentralized control. Information-theoretically, the problem is closely related to information embedding problems studied, for instance, in~\cite{CostaDirtyPaper,SutivongGaussian,KimStateAmplification,KotagiriLaneman,MerhavMasking,KoetterDPC}, and more closely in~\cite{ReversibleStegotext}, etc. These are variations on the work of Gel'fand and Pinsker~\cite{gelfandpinsker}, where the decoder is interested in reconstructing a message embedded in the state by the transmitter. Most of these variations (e.g.~\cite{SutivongGaussian,KimStateAmplification,MerhavMasking,KoetterDPC}) require the communication or hiding of the \textit{original} host signal: for instance, the work of Sutivong \textit{et al}.~\cite{SutivongGaussian,KimStateAmplification} addresses the problem of imperfect reconstruction of the \textit{original} host signal $\m{X}_0$, while simultaneously communicating a message.%In~\cite{KimSutivongCover,KimStateAmplification}, the authors introduce a further constraint of reconstructing the host signal (within some distortion) along with the message. 

% In their information embedding problem, the encoder ensures that the decoder recovers the \textit{modified} host signal $\xonevec $ perfectly, along with the message that is embedded in $\m{X}_1$. Intellectually, the work in~\cite{ReversibleStegotext} is directed towards understanding how a communication problem changes when an additional requirement, that of the encoder being able to produce a copy of the decoder's reconstruction, is imposed on the system\footnote{The authors in~\cite{ReversibleStegotext} allow the host signal to modify the channel directly as well, which we do not allow in our formulation. This is why our extension is only partial.} (in a source coding context, the issue was explored by Steinberg in~\cite{SteinbergHostRecovery,SteinbergJournal}).%all of these works focus on reconstruction or hiding of the \textit{original host signal}. %We refer the interested reader to~\cite{WitsenhausenJournal}, where these connections are discussed in detail. 

\begin{figure}[htb]
\begin{center}
\includegraphics[scale=0.4]{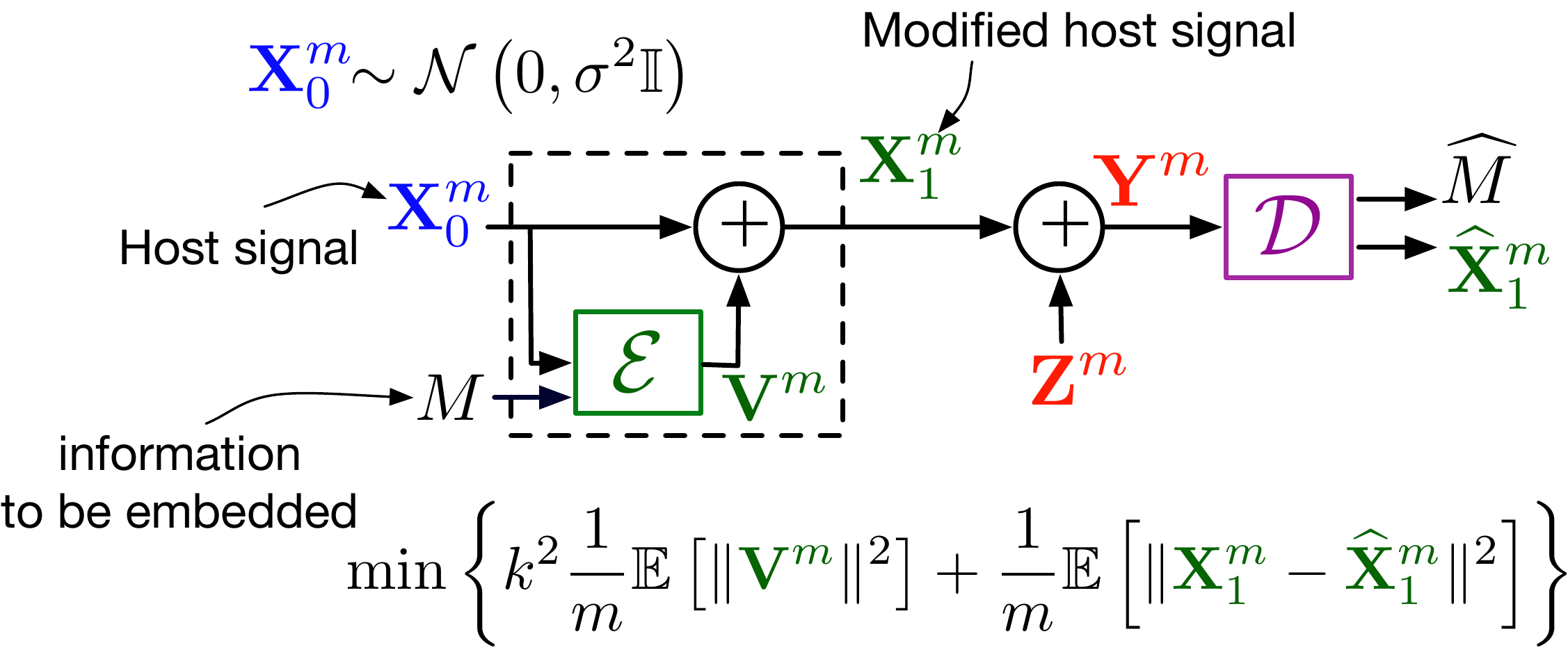}
\caption{The host signal $\instvec$ is first modified by the encoder $\mathcal{E}$ using a power-constrained input $\m{V}$. The modified host signal $\xonevec $ and the message $M$ are then reconstructed at the decoder. The problem is to find the required cost to attain rate $R$, where the cost is a weighted sum of the average power $P$ used to modify the host signal and the average distortion in the reconstruction of $\xonevec $. }
\label{fig:blockdgm}
\end{center}
\end{figure}

In comparison, our problem requires recovery of the \textit{modified} host signal. Our problem is a partial extension of an information-embedding problem considered by Sumszyk and Steinberg~\cite{ReversibleStegotext}. The authors in~\cite{ReversibleStegotext} focus on the finite-alphabet case, and ask the question of the achievable message-rates for the case where the decoder recovers the embedded message as well as the \textit{modified} host signal $\xonevec $ perfectly. Intellectually, the work in~\cite{ReversibleStegotext} is directed towards understanding how a communication problem changes when an additional requirement, that of the encoder being able to produce a copy of the decoder's reconstruction, is imposed on the system (in a source coding context, the issue was explored by Steinberg in~\cite{SteinbergHostRecovery,SteinbergJournal}). %all of these works focus on reconstruction or hiding of the \textit{original host signal}. %We refer the interested reader to~\cite{WitsenhausenJournal}, where these connections are discussed in detai

In this paper, we partially  extend Sumszyk and Steinberg's problem to the Gaussian case\footnote{The authors in~\cite{ReversibleStegotext} allow the channel to be \textit{directly} affected by the original host signal as well as the modified host. In our formulation, the channel is affected by the modified host, but not the original host. This is why our extension is only partial.} (the formal problem statement is in Section~\ref{sec:probstat}). Further, unlike~\cite{ReversibleStegotext}, we allow \textit{imperfect} recovery of the modified host, although our results also characterize the optimal rate in the limit of perfect recovery. Often in infinite-alphabet source-coding problems, perfect reconstruction across finite-capacity channels is impossible. However, it turns out that for this problem of recovering the modified host signal $\m{X}_1$, the reconstruction \textit{can} be perfect. For instance, the encoder can ensure that the modified host signal takes values in a discrete subset of the continuous space. 
%\footnote{Experience in infinite-alphabet source-coding problems might even suggest that  (asymptotic) perfect reconstruction may be impossible because the modified host signal $\m{X}_1$ (which is a ``message'' of communication), can take values in a continuous space. This is not the case here. For instance, the encoder could simply force the host signal to zero.}.% simple Fano's inequality-based techniques do not directly apply%

When could one be interested in recovering a modified version of the host-signal? Such situations often arise in problems of control, where the agents are interested in estimating the \textit{current} state of the system (post-modification), and may not care about the original state\footnote{Sumszyk and Steinberg's motivation appears to be closely related. They state~\cite{ReversibleStegotext}: ``if the user is interested in decoding the data and further transmitting the stegotext, \textit{i.e}, the result of the embedding of independent data into the host, then full decoding of the host is not needed, only re-production of the stegotext.''}. A simple decentralized-control example of this feature is the celebrated ``counterexample'' of Witsenhausen~\cite{Witsenhausen68}, shown in Fig.~\ref{fig:Wcounter}(a). Formulated in 1968, the goal in the scalar problem is to minimize a weighted sum\footnote{It is shown in~\cite{WitsenhausenJournal} that obtaining the optimal weighted sum cost for all weights is equivalent to obtaining the tradeoff between the power and distortion costs. } of power required to modify the original host signal $X_0$, and the mean-square error in reconstructing the modified host $X_1$ (there is no explicit message to communicate in the problem). The problem shows that the ``certainty-equivalence'' doctrine, which has proven extremely useful in centralized control~\cite{CertaintyEquivalence}, breaks down when applied to decentralized control. This doctrine suggests a separation of estimation and control\footnote{This separation is conceptually analogous to the source-channel separation theorem in information theory~\cite{ShannonOriginalPaper}, which also breaks down on increased decentralization, \textit{i.e.} in large networks~\cite{CoverElGamalSalehi}.}, which can be proven to be optimal in many centralized cases~\cite{WitsenhausenPatterns}. The separation fails to hold for Witsenhausen's counterexample (which is a decentralized control problem) because besides reducing immediate control costs, the control actions can also improve the \textit{estimability} of the modified host signal $\m{X}_1$. This is often called the ``dual role'' of control actions in control theory~\cite{BarShalomTse}.

\begin{figure}[htbp] %  figure placement: here, top, bottom, or page
   \centering
   \includegraphics[width=4.7in]{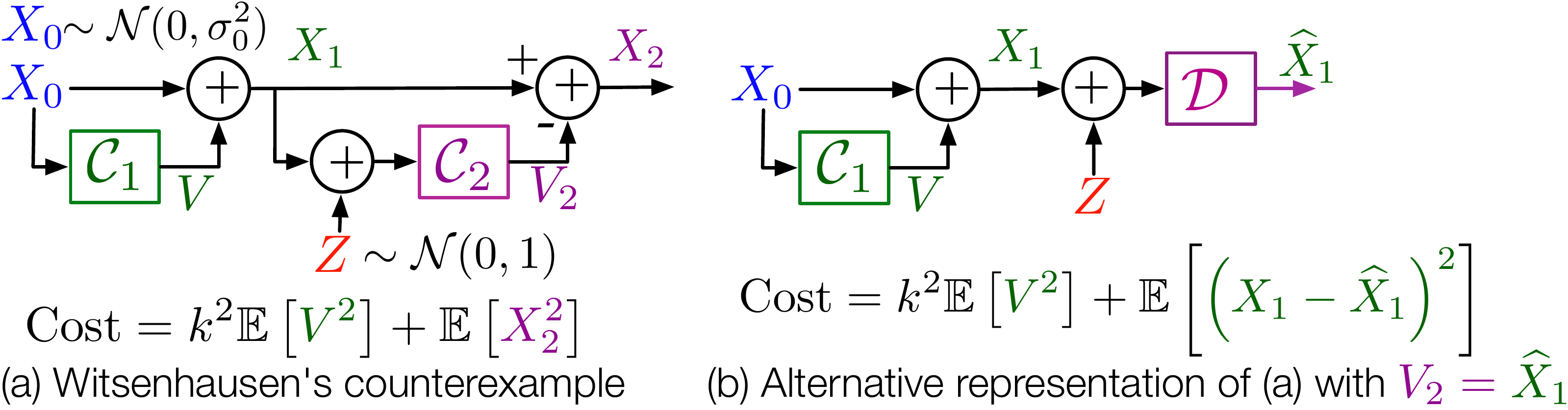} 
   \caption{From Witsenhausen's counterexample to a scalar version of the information-embedding problem. The problem of finding optimal costs for the vector version of (b) is the same as that of finding the optimal tradeoff between the power and MMSE costs. }
   \label{fig:Wcounter}
\end{figure}

In our problem, the encoder's actions need to balance \emph{three} roles:
(1) to minimize the transmitted power, (2) to make it easier for the decoder
to estimate the host signal, and (3) to communicate with the decoder.
The tension between these same three tasks also arises in team
decision theory~\cite{HoTeamDecision}. There agents wish to choose their actions in order
to (1) minimize some cost function, (2) help other agents to
estimate a hidden state and (3) send messages to other agents to
help them all to coordinate their actions. In the stochastic control
literature, this is sometimes called the ``triple role of control''~\cite{HoTeamDecision}.
Multiple roles of control actions have long been of
interest in control theory (e.g.~\cite{BarShalomTse,VaraiyaBook,Feldbaum}), and have recently
been investigated in information theory by Cuff and Zhao~\cite{CuffITW11}.

In Section~\ref{sec:main}, we characterize the optimal control costs to within a factor of $16$ for \textit{all message rates} uniformly over all problem parameters. Our technical contribution is two-fold: in achievability and converse. From an achievability perspective, our dirty-paper coding-based strategy requires recovering the auxiliary random variable along with the message-bin it lies in (unlike in~\cite{gelfandpinsker,CostaDirtyPaper} and related works, where only the message-bin needs to be recovered). This introduces a dependence between non-transmitted codewords and channel outputs that requires us to go beyond the standard typicality-based techniques in related literature\footnote{We are extremely grateful to the anonymous reviewer who pointed out this difficulty to us. While it has been acknowledged and addressed in some settings in information-theoretic literature (e.g.~\cite{LapidothTinguely,AaronKelly}), it has been ignored in others, e.g.~\cite{ReversibleStegotext,KimStateAmplification,KotagiriLaneman} and our own work~\cite{WitsenhausenJournal}. In comparison with~\cite{LapidothTinguely}, our proofs are significantly shorter, and in comparison with~\cite{AaronKelly}, which focuses on finite-alphabet channels, our proofs work directly for the infinite-alphabet Gaussian case. }.

From a converse perspective, our results tighten the known bounds on the Witsenhausen problem. The new lower bound in this paper specialized to zero-rate case provides an improved lower bound to the costs of the vector Witsenhausen counterexample in the long-blocklength limit. Using this improved bound, we numerically show that the ratio of upper and lower bounds is smaller than $1.3$ regardless of the choice of the weights and the problem parameters. This is tighter than the previously best known (numerical) bound of~\cite{WitsenhausenJournal}, which provided a larger ratio of 2.

Interestingly, we show in Theorem~\ref{thm:match}  that the \textit{optimal} strategy for asymptotically perfect reconstruction for communication at any rate $R$ is a dirty-paper-coding-based strategy. In comparison, the optimal strategies for many seemingly-similar problems (e.g.~\cite{KimStateAmplification,SutivongGaussian,MerhavMasking}) that recover the original host signal (as opposed to the modified host signal) in the zero-rate case are linear. That dirty-paper coding is needed for our problem even in the zero-rate case suggests that a deeper study is required to understand the intellectual difference between problems that recover the original host vis-a-vis problems that recover the modified host.

\section{Problem Statement}
\label{sec:probstat}
The host signal $\instvec$ is distributed $\mathcal{N}(0,\sigma^2\mathbb{I})$, and the message $M$ is independent of $\instvec$ and distributed uniformly over $\{1,2,\ldots,2^{mR}\}$. The encoder uses map $\mathcal{E}_m$ to map $(M,\instvec)$ to $\xonevec $ by additively distorting $\instvec$ using input $\m{V}$ of average power (for each message) at most $P$, i.e. $\expect{\|\instvec-\xonevec \|^2}\leq mP$. Additive white Gaussian noise $\m{Z}\sim\mathcal{N}(0,\sigma_z^2\mathbb{I})$, where $\sigma_z^2=1$, is added to $\xonevec $ by the channel. The decoder $\mathcal{D}_m$ maps the channel output $\m{Y}$ to both an estimate $\m{\widehat{X}}_1$ of the modified host signal $\xonevec $ and an estimate $\widehat{M}$ of the message. 

Define the error probability $\epsilon_m(\mathcal{E}_m,\mathcal{D}_m)=\Pr(M\neq \widehat{M})$. Define the distortion function $d(x,y)=(x-y)^2$ for $x,y\in\mathbb{R}$, and for $\m{x},\m{y}\in\mathbb{R}^m$, $d(\m{x},\m{y})=\frac{1}{m}\|\m{x}-\m{y}\|^2$.  For the map sequence $\{\mathcal{E}_m,\mathcal{D}_m\}_{m=1}^\infty$, define the minimum asymptotic distortion $\mathrm{MMSE}(P,R)$ as follows
\begin{eqnarray*}
\mathrm{MMSE}(P,R) &:=& \underset{\{\mathcal{E}_m,\mathcal{D}_m\}_{m=1}^\infty : \epsilon_m(\mathcal{E}_m,\mathcal{D}_m)\rightarrow 0}{\inf} \;\underset{m\rightarrow\infty}{\lim \sup}\; \expect{d(\m{X}_1,\m{\widehat{X}}_1)}\\
&=& \underset{\{\mathcal{E}_m,\mathcal{D}_m\}_{m=1}^\infty : \epsilon_m(\mathcal{E}_m,\mathcal{D}_m)\rightarrow 0}{\inf} \;\underset{m\rightarrow\infty}{\lim \sup}\; \frac{1}{m}\expect{\|\xonevec -\m{\widehat{X}}_1\|^2}.
\end{eqnarray*}
Our results focus on the tradeoff between the rate $R$, the power $P$, and $\mathrm{MMSE}(P,R)$.

We will also use our results to obtain bounds on the conventional control-theoretic weighted cost~\cite{Witsenhausen68}:
\begin{equation}
\cost = \frac{1}{m}k^2 \|\m{V}\|^2 + \frac{1}{m}\|\xonevec -\m{\widehat{X}}_1\|^2,
\end{equation}
where $k\in\mathbb{R}^+$. The objective is to minimize the average cost, $\avgcost:=\expect{J}$ at rate $R$. The optimal average cost is denoted by $\avgcost_{\text{opt}}$. The average is taken over the realizations of the host signal, the channel noise, and the message. At $R=0$, the problem is the vector Witsenhausen counterexample~\cite{WitsenhausenJournal}. 

%%%%%%%%%%%%%%%%%%%%%%%%%%%%%%%%%%%%
\section{Main Results}
\label{sec:main}
\subsection{Lower bounds on $\mathrm{MMSE}(P,R)$}

%\subsubsection{Why Fano's inequality is not sufficient to obtain lower bounds $\mathrm{MMSE}(P,R)$}

%
%
%
%Because the set $\m{\mathcal{X}}\subseteq \mathbb{R}^m$ depends on the choice of control policy [...]. 
%
%For any policy that has a bounded support for the set $\m{\mathcal{X}}$, let $\epsilon$ be the minimum distance between any two points. Then draw spheres of radius $\frac{\epsilon}{4m}$ in this space. These spheres would be disjoint. As $m\rightarrow\infty$, one can reconstruct which sphere the 
%
%There is another way in which one can have asymptotically perfect reconstruction of $\m{X}_1$ without 
%
%
%My argument seems to be: for any policy $\gamma$, I can construct another policy $\gamma'$ of infinite alphabet size that requires asymptotically the same power for perfect reconstruction. Clearly, Fano's inequality does not work for this alternative policy which has the same performance. Discretizing this policy wouldn't work either. Why?
%
%Fano-style arguments don't work when we are dealing with continuous/infinite-alphabet sources. 

%I'll probably need to argue that even the proof of Fano's inequality doesn't work.

%\subsubsection{The lower bounds}

\begin{theorem}
\label{thm:newlower}
For the problem as stated in Section~\ref{sec:probstat}, for communicating reliably at rate $R$ with input power $P$, the asymptotic average mean-square error in recovering $\xonevec $ is lower bounded as follows. For $P\geq 2^{2R}-1$,
\begin{equation}
\label{eq:newLowerWithRateThm}
\mathrm{MMSE}(P,R)\geq \inf_{\corr}\sup_{\gamma\in\mathbb{R}} \frac{1}{\gamma^2}\left(\left( \sqrt{\frac{\sigma^2 2^{2R}}{1+\sigma^2+P+2\corr}}  -   \sqrt{(1-\gamma)^2\sigma^2 + \gamma^2 P -2\gamma (1-\gamma)\corr}     \right)^+\right)^2,
\end{equation}
where $\max\left\{-\sigma\sqrt{P}, \frac{2^{2R}-1-P -\sigma^2}{2}\right\}   \leq\corr\leq\sigma\sqrt{P}$ is the range of $\sigma_{XV}$ over which the above infimum is being taken. For $P<2^{2R}-1$, reliable communication at rate $R$ is not possible. 
\end{theorem}

\begin{corollary}\label{coro:wit}
For the vector Witsenhausen problem (\textit{i.e.} $R=0$ case) with input power $P$, the following is a lower bound on the $\mathrm{MMSE}$ in the estimation of $\xonevec $:
\begin{equation}
\label{eq:newLower}
\mathrm{MMSE}(P,0)\geq \inf_{\corr}\sup_{\gamma\in\mathbb{R}} \frac{1}{\gamma^2}\left(\left( \sqrt{\frac{\sigma^2}{1+\sigma^2+P+2\corr}}  -   \sqrt{(1-\gamma)^2\sigma^2 + \gamma^2 P -2\gamma (1-\gamma)\corr}     \right)^+\right)^2,
\end{equation}
where $\corr\in [-\sigma\sqrt{P},\sigma\sqrt{P}]$. Further, this lower bound on $\expect{d(\m{X}_1,\whatmn{X}_1)}$ holds for all $m$, and not just asymptotically.
\end{corollary}
\begin{proof}\textit{[Of Theorem~\ref{thm:newlower}]}
For conceptual clarity, we first derive the result for the case $R=0$ (Corollary~\ref{coro:wit}). The tools developed are then used to derive the lower bound for $R> 0$. 
 
\begin{proof}\textit{[Of Corollary~\ref{coro:wit}]}

For any chosen pair of control maps $\mathcal{E}_m$ and $\mathcal{D}_m$, there is a Markov chain $\instvec\rightarrow \xonevec \rightarrow \m{Y}\rightarrow \m{\widehat{X}}_1$. Using the data-processing inequality
\begin{equation}
\label{eq:dpi}
I(\instvec;\m{\widehat{X}}_1) \leq I(\xonevec ;\m{Y}).
\end{equation}
The terms in the inequality can be bounded by single letter expressions as follows. Define $Q$ as a random variable uniformly distributed over $\{1,2,\ldots,m\}$. Define $X_0=X_{0,Q}$, $V=V_Q$, $X_1=X_{1,Q}$, $Z=Z_Q$, $Y=Y_Q$ and $\widehat{X}_1=\widehat{X}_{1,Q}$. Then, 
\begin{eqnarray}
I(\xonevec ;\m{Y}) &=& h(\m{Y}) - h(\m{Y}|\xonevec )\nonumber\\
&\overset{(a)}{\leq}& \sum_ih(Y_i) - h(\m{Y}|\xonevec )= \sum_ih(Y_i) - h(Y_i|X_{1,i})= \sum_iI(X_{1,i};Y_i)= m I(X_1;Y|Q)\nonumber\\
&= & m\left( h(Y|Q) - h(Y|X_1,Q)\right)\leq  m\left( h(Y) - h(Y|X_1,Q)\right)\nonumber\\
&\overset{(b)}{=} &m\left( h(Y) - h(Y|X_1)\right) = m I(X_1;Y),
\label{eq:ixy}
\end{eqnarray}
where $(a)$ follows from an application of the chain-rule for entropy followed by using the fact that conditioning reduces entropy, and $(b)$ follows from the observation that the additive noise $Z_i$ is iid across time, and independent of the input $X_{1,i}$ (thus $Y$ is conditionally independent of  $Q$ given $X_1$). Also, 
\begin{eqnarray}
I(\instvec;\m{\widehat{X}}_1) &=& h(\instvec) - h(\instvec|\m{\widehat{X}}_1) =  \sum_i h(X_{0,i}) - h(\instvec|\m{\widehat{X}}_1) \nonumber\\
& \overset{(a)}{\geq} &  \sum_i \left(h(X_{0,i}) - h(X_{0,i}|\widehat{X}_{1,i})\right) \nonumber\\
& = &  \sum_i I(X_{0,i};\widehat{X}_{1,i}) = m I(X_0;\widehat{X}_1|Q)\nonumber= m\left(h(X_0|Q) - h(X_0|\widehat{X}_1,Q)\right)\nonumber\\
& \overset{(b)}{\geq} & m\left( h(X_0) - h(X_0|\widehat{X}_1)\right) = mI(X_0;\widehat{X}_1),
\label{eq:isx}
\end{eqnarray}
where $(a)$ and $(b)$ again follow from the fact that conditioning reduces entropy, and $(b)$ also uses the observation that since $X_{0,i}$ are iid, $X_0$, $X_{0,i}$, and $X_0|Q=q$ are distributed identically.

Now, using~\eqref{eq:dpi},~\eqref{eq:ixy} and~\eqref{eq:isx},
\begin{equation}
\label{eq:bigineq}
mI(X_0;\widehat{X}_1)\leq I(\instvec;\m{\widehat{X}}_1) \leq I(\xonevec ;\m{Y}) \leq mI(X_1;Y).
\end{equation}
Also observe that from the definitions of $X_1$, $\widehat{X}_1$ and $Y$, $\expect{d(\instvec,\xonevec )}=\expectp{\m{X}_0}{\frac{1}{m}\sum_{i=1}^md(X_{0,i},X_{1,i})}=\expectp{\m{X}_0}{\expectp{Q}{d(X_0,X_1)}|\m{X}_0}=\expect{d(X_0,X_1)}$, and similarly, $\expect{d(\xonevec ,\m{\widehat{X}}_1)}=\expect{d(X_1,\widehat{X}_1)}$. 

Using the Cauchy-Schwartz inequality, the correlation $\corr=\expect{X_0V}$  must satisfy the following constraint,
\begin{equation}
\label{eq:sigmasu}
|\corr|= |\expect{X_0V}|\leq \sqrt{\expect{X_0^2}}\sqrt{\expect{V^2}}\leq \sigma\sqrt{P}.
\end{equation}
Also,
\begin{equation}
\label{eq:xpow}
\expect{X_1^2}=\expect{(X_0+V)^2}= \sigma^2 + P + 2\sigma_{XV}.
\end{equation} 
Since $Z=Y-X_1$ is independent of $X_1$, and a Gaussian input distribution maximizes the mutual information across an average-power-constrained AWGN channel,
\begin{equation}
\label{eq:ixycompute}
I(X_1;Y)\leq \frac{1}{2}\lo{1+\frac{P+\sigma^2+2\corr}{1}}.
\end{equation}
\begin{eqnarray}
I(X_0;\widehat{X}_1) &=& h(X_0) - h(X_0|\widehat{X}_1)\nonumber\\
& = & h(X_0) - h(X_0-\gamma \widehat{X}_1|\widehat{X}_1)\;\forall \gamma \nonumber\\
& \overset{(a)}{\geq} & h(X_0) - h(X_0-\gamma \widehat{X}_1)\nonumber\\
& = & \frac{1}{2}\lo{2\pi e\sigma^2} - h(X_0-\gamma \widehat{X}_1),
\label{eq:lossy}
\end{eqnarray}
where $(a)$ follows from the fact that conditioning reduces entropy. Also note here that the result holds for any $\gamma$, and in particular, $\gamma$ can depend on $\corr$. Now, 
\begin{eqnarray}
\label{eq:sminusxhat}
 h(X_0-\gamma\widehat{X}_1) &= & h(X_0-\gamma(\widehat{X}_1-X_1) -\gamma X_1)= \textcolor{red}{h\left(X_0 -\gamma (\widehat{X}_1-X_1)- \gamma X_0 -\gamma V \right)}\nonumber \\
& = & h\left( (1-\gamma)X_0-\gamma V -\gamma (\widehat{X}_1-X_1) \right).
\end{eqnarray}
The second moment of a sum of two random variables $A$ and $B$ can be bounded as follows
\begin{eqnarray}
\label{eq:aplusbsq}
\expect{(A+B)^2} & = & \expect{A^2}+\expect{B^2}+2\expect{AB}\nonumber\\
& \overset{\text{Cauchy-Schwartz ineq.}}{\leq} & \expect{A^2}+\expect{B^2}+2\sqrt{\expect{A^2}}\sqrt{\expect{B^2}}\nonumber\\
& = & (\sqrt{\expect{A^2}}  + \sqrt{\expect{B^2}})^2,
\end{eqnarray}
with equality when $A$ and $B$ are aligned, i.e. $A=\lambda B$ for some $\lambda\in \mathbb{R}$. For the random variables under consideration in~\eqref{eq:sminusxhat}, choosing $A = (1-\gamma)X_0-\gamma V$, and $B = -\gamma (\widehat{X}_1-X_1)$ in~\eqref{eq:aplusbsq}
\begin{eqnarray}
\label{eq:align}
&&\nonumber\expect{\left((1-\gamma)X_0 - \gamma V -\gamma (\widehat{X}_1-X_1) \right)^2 } \\
&&\leq \left(\sqrt{(1-\gamma)^2\sigma^2 + \gamma^2 P -2\gamma (1-\gamma)\corr} + |\gamma| \sqrt{\expect{(\widehat{X}_1-X_1)^2}}\right)^2.
\end{eqnarray}
Equality is obtained by aligning\footnote{In general, since $\m{\widehat{X}}_1$ is a function of $\m{Y}$, this alignment is not actually possible when the recovery of $\xonevec $ is not exact. This is one reason why the derived bound is loose for inexact reconstruction.} $X_1-\widehat{X}_1$ with $(1-\gamma)X_0 - \gamma V$. Thus, from~\eqref{eq:lossy},~\eqref{eq:sminusxhat}, and~\eqref{eq:align},
\begin{eqnarray}
I(X_0;\widehat{X}_1) &\geq &\frac{1}{2}\lo{2\pi e\sigma^2} - h(X_0-\gamma\widehat{X}_1)\nonumber\\
&\geq & \frac{1}{2}\lo{\frac{\sigma^2}{\left(\sqrt{(1-\gamma)^2\sigma^2 + \gamma^2 P -2\gamma (1-\gamma)\corr} + |\gamma| \sqrt{\expect{(\widehat{X}_1-X_1)^2}}\right)^2}}.
\label{eq:isxhatcompute}
\end{eqnarray}
From~\eqref{eq:bigineq}, $I(X_0;\widehat{X}_1)\leq I(X_1;Y)$. Using the lower bound on $I(X_0;\widehat{X}_1)$ from~\eqref{eq:isxhatcompute} and the upper bound on $I(X_1;Y)$ from~\eqref{eq:ixycompute}, we get
$$
\frac{1}{2}\lo{\frac{\sigma^2}{\left(\sqrt{(1-\gamma)^2\sigma^2 + \gamma^2 P -2\gamma (1-\gamma)\corr} + |\gamma| \sqrt{\expect{(\widehat{X}_1-X_1)^2}}\right)^2}} \leq \frac{1}{2}\lo{1+\frac{P+\sigma^2+2\corr}{1}},
$$
for the choice of $\mathcal{E}_m$ and $\mathcal{D}_m$. Since $\lo{\cdot{}}$ is a monotonically increasing function, 
\begin{eqnarray*}
\frac{\sigma^2}{\left(\sqrt{(1-\gamma)^2\sigma^2 + \gamma^2 P -2\gamma (1-\gamma)\corr} + |\gamma| \sqrt{\expect{(\widehat{X}_1-X_1)^2}}\right)^2}\leq 1+P+\sigma^2+2\corr\\
i.e.,\;\;\left(\sqrt{(1-\gamma)^2\sigma^2 + \gamma^2 P -2\gamma (1-\gamma)\corr} + |\gamma| \sqrt{\expect{(\widehat{X}_1-X_1)^2}}\right)^2\geq \frac{\sigma^2}{1+P+\sigma^2+2\corr},\\
i.e.,\;\;|\gamma| \sqrt{\expect{(\widehat{X}_1-X_1)^2}}\geq \sqrt{\frac{\sigma^2}{1+P+\sigma^2+2\corr}}- \sqrt{(1-\gamma)^2\sigma^2 + \gamma^2 P -2\gamma (1-\gamma)\corr}. 
\end{eqnarray*}
Because the RHS may not be positive, we take the maximum of zero and the RHS and obtain the following lower bound for $\mathcal{E}_m$ and $\mathcal{D}_m$. 
\begin{equation}
\expect{(\widehat{X}_1-X_1)^2}\geq  \frac{1}{\gamma^2}\left(\left(\sqrt{\frac{\sigma^2}{1+P+\sigma^2+2\corr}}- \sqrt{(1-\gamma)^2\sigma^2 + \gamma^2 P -2\gamma (1-\gamma)\corr}\right)^+\right)^2.
\end{equation}
Because the bound holds for every $\gamma$,
\begin{equation}
\expect{(\widehat{X}_1-X_1)^2}\geq \sup_{\gamma\in\mathbb{R}} \frac{1}{\gamma^2}\left(\left(\sqrt{\frac{\sigma^2}{1+P+\sigma^2+2\corr}}- \sqrt{(1-\gamma)^2\sigma^2 + \gamma^2 P -2\gamma (1-\gamma)\corr}\right)^+\right)^2,
\end{equation}
for the chosen $\mathcal{E}_m$ and $\mathcal{D}_m$. %There is no loss in maximizing over $\gamma>0$ (as compared to maximizing over all $\gamma$) because, as is easy to check, the bound is maximized for some $\gamma>0$.

Now, from~\eqref{eq:sigmasu}, $\corr$ can take values in $[-\sigma\sqrt{P},\sigma\sqrt{P}]$. Because the lower bound depends on $\mathcal{E}_m$ and $\mathcal{D}_m$ only through $\corr$, we obtain the following lower bound for all $\mathcal{E}_m$ and $\mathcal{D}_m$,
\begin{equation}
\expect{(\widehat{X}_1-X_1)^2}\geq \inf_{|\corr|\leq \sigma\sqrt{P}}\sup_{\gamma\in\mathbb{R}} \frac{1}{\gamma^2}\left(\left(\sqrt{\frac{\sigma^2}{1+P+\sigma^2+2\corr}}- \sqrt{(1-\gamma)^2\sigma^2 + \gamma^2 P -2\gamma (1-\gamma)\corr}\right)^+\right)^2,
\end{equation}
which proves Corollary~\ref{coro:wit}. Notice that we did not take limits in ``blocklength'' $m$ anywhere, and hence the lower bound here holds for all values of $m$.
\end{proof}

\subsection*{The case of nonzero rate}
To prove Theorem~\ref{thm:newlower}, consider now the problem when the encoder wants to also communicate a message $M$ reliably to the decoder at rate $R$. 
%In the following, we use Fano's inequality in order to arrive at the necessary conditions to reconstruct the message $M$. For reasons discussed in Section~\ref{sec:intro}, Fano's inequality can still not be used for arriving at necessary conditions for reconstructing $\m{X}_1$, but our derivation of lower bounds on mean-square error in reconstructing $\m{X}_1$ in the proof of Corollary~\ref{coro:wit} will prove useful.
Since reliable communication requires $\Pr(M\neq \widehat{M})=\epsilon_m\rightarrow 0$ as $m\rightarrow\infty$, using Fano's inequality, $H(M|\widehat{M})\leq m\delta_m$ where $\delta_m\rightarrow 0$. Thus,
\begin{eqnarray}
\nonumber I(M;\widehat{M})&=& H(M)-H(M|\widehat{M})= mR - H(M|\widehat{M})\\
&\geq & mR - m\delta_m = m(R-\delta_m).
\label{eq:fano}
\end{eqnarray}
As before, we consider a mutual information inequality that follows directly from the Markov chain $(M,\instvec)\rightarrow \xonevec \rightarrow \m{Y}\rightarrow (\m{\widehat{X}}_1,\widehat{M})$ :
\begin{equation}
\label{eq:dpi2}
I(M,\instvec;\widehat{M},\m{\widehat{X}}_1) \leq I(\xonevec ;\m{Y}).
\end{equation}
The RHS can be bounded above as in~\eqref{eq:ixy}. For the LHS,
\begin{eqnarray}
\nonumber 
I(M,\instvec;\widehat{M},\m{\widehat{X}}_1) &=& I(M;\widehat{M},\m{\widehat{X}}_1) + I(\instvec;\widehat{M},\m{\widehat{X}}_1|M)\\
\nonumber & \geq &I(M;\widehat{M}) +   I(\instvec;\widehat{M},\m{\widehat{X}}_1|M)\\
\nonumber & = &I(M;\widehat{M}) +  h(\instvec|M) -h(\instvec|\widehat{M},\m{\widehat{X}}_1,M)\\
\nonumber & \overset{\instvec\;\text{indep. of}\; M}{=} &I(M;\widehat{M}) +  h(\instvec) -h(\instvec|\widehat{M},\m{\widehat{X}}_1,M)\\
\nonumber  &\geq &I(M;\widehat{M}) +  h(\instvec) -h(\instvec|\m{\widehat{X}}_1)=I(M;\widehat{M}) +  I(\instvec;\m{\widehat{X}}_1)\nonumber\\
&\overset{\text{using~\eqref{eq:isx}}}{\geq} &I(M;\widehat{M}) +  mI(X_0;\widehat{X}_1).
\label{eq:breaking}
\end{eqnarray}
From~\eqref{eq:fano},~\eqref{eq:dpi2} and~\eqref{eq:breaking}, we obtain
\begin{eqnarray}
\label{eq:bigineq2}
\nonumber m(R-\delta_m)+  mI(X_0;\widehat{X}_1)& \overset{\text{using}~\eqref{eq:fano}}{\leq} &I(M;\widehat{M}) +  mI(X_0;\widehat{X}_1)\\
\nonumber & \overset{\text{using}~\eqref{eq:breaking}}{\leq} &I(M,\instvec;\widehat{M},\m{\widehat{X}}_1) \\
 &\overset{\text{using}~\eqref{eq:dpi2}}{\leq} & I(\xonevec ;\m{Y})\overset{\text{using}~\eqref{eq:ixy}}{\leq} m I(X_1;Y).
\end{eqnarray}
$I(X_1;Y)$ and $I(X_0;\widehat{X}_1)$ can be bounded as before in~\eqref{eq:ixycompute} and~\eqref{eq:isxhatcompute}. Observing that as $m\rightarrow\infty$, $\delta_m\rightarrow 0$, we get the following lower bound on the $\mathrm{MMSE}$ for nonzero rate,
\begin{equation}
\label{eq:newLowerWithRate}
\mathrm{MMSE}(P,R)\geq \inf_{\corr}\sup_{\gamma\in\mathbb{R}} \frac{1}{\gamma^2}\left(\left( \sqrt{\frac{\sigma^2 2^{2R}}{1+\sigma^2+P+2\corr}}  -   \sqrt{(1-\gamma)^2\sigma^2 + \gamma^2 P -2\gamma (1-\gamma)\corr}     \right)^+\right)^2.
\end{equation}
In the limit $\delta_m\rightarrow 0$, we require from~\eqref{eq:bigineq2} that $I(X_1;Y)\geq R$. This gives the following constraint on $\corr$,
\begin{eqnarray}
\frac{1}{2}\lo{1+P+\sigma^2+2\corr}\geq R \nonumber \\
\text{i.e.}\; \corr\geq \frac{2^{2R}-1-P-\sigma^2}{2},
\end{eqnarray}
yielding (in conjunction with~\eqref{eq:sigmasu}) the constraint on $\corr$ in Theorem~\ref{thm:newlower}. The constraint on $P$ in the Theorem follows from Costa's result~\cite{CostaDirtyPaper}, because the rate $R$ must be smaller than the capacity over a power constrained AWGN channel with known interference, $\frac{1}{2}\lo{1+P}$. 
\end{proof}
It is insightful to see how the lower bound in Corollary~\ref{coro:wit} is an improvement over that in~\cite{WitsenhausenJournal}. The lower bound in~\cite{WitsenhausenJournal} is given by
\begin{equation}
\label{eq:oldbounddd}
\mathrm{MMSE}(P,0)\geq \left(\left(  \sqrt{\frac{\sigma^2}{(\sigma+\sqrt{P})^2+1}}   -\sqrt{P}\right)^+\right)^2,
\end{equation}
which again holds for all $m$. Because any $\gamma$ provides a valid lower bound in Corollary~\ref{coro:wit}, choosing $\gamma=1$ in the bound of Theorem~\ref{thm:newlower} provides the following (loosened) bound,
\begin{equation}
\label{eq:rbd}
\mathrm{MMSE}(P,R)\geq \inf_{|\corr|\leq \sigma\sqrt{P}}\left(\left(  \sqrt{\frac{\sigma^2 2^{2R}}{\sigma^2+P+2\corr+1}}   -\sqrt{P}\right)^+\right)^2.
\end{equation}
The infimum in this loosened bound is achieved at $\corr=\sigma\sqrt{P}$, yielding a new loosened bound\footnote{This loosened bound proves convenient for algebraic manipulations in the proof of Theorem~\ref{thm:approximate}.} on $\mathrm{MMSE}(P,R)$,
\begin{equation}
\label{eq:looselower}
\mathrm{MMSE}(P,R)\geq\left(\left(  \sqrt{\frac{\sigma^2 2^{2R}}{\sigma^2+P+2\sigma\sqrt{P}+1}}   -\sqrt{P}\right)^+\right)^2,
\end{equation}
matching the expression in~\eqref{eq:oldbounddd} at $R=0$.

%%%%%%%%%%%%%%%%%%%%%%%%%%%%%%%%%%%%%%%%%

\subsection{The upper bound}
\label{sec:dpc}
We use the combination of linear and dirty-paper coding strategies of~\cite{WitsenhausenJournal}, but now we communicate a message at rate $R$ as well. %The analysis of the strategy in the next section borrows results from~\cite{CostaDirtyPaper}.%We summarize the strategy briefly, and refer the interested reader to~\cite{WitsenhausenJournal} for a detailed description and analysis of the achievability. 

\subsubsection{A description of the combination of linear and DPC-based strategy}
\label{sec:dpcdes}
The encoder divides its input into two parts $\m{V}=\m{V}_{\mathrm{lin}}+\m{V}_{\mathrm{dpc}}$, where $\m{V}_{\mathrm{lin}}$ and $\m{V}_{\mathrm{dpc}}$ have powers $P_{\mathrm{lin}}$ and $P_{\mathrm{dpc}}$ respectively.  %We refer to $P_{\mathrm{lin}}$ as the \textit{linear} part of the power, and $P_{\mathrm{dpc}}$ the \textit{dirty-paper coding} part of the power. 
The linear part $\m{V}_{\mathrm{lin}}$ is used to scale the host signal down by a factor $ (1-\beta)$ (using $\m{V}_{\mathrm{lin}}=-\beta\m{X}_0$) so that $P_{\mathrm{lin}}= \beta^2\sigma^2$. The resulting scaled down host signal $\wtildemn{X}_0=(1-\beta)\m{X}_0$ has variance $\widetilde{\sigma}^2=\sigma^2(1-\beta)^2$ in each dimension.

The remaining ``interference'' is $\wtildemn{X}_0$, which has variance $\widetilde{\sigma}^2$. The transmitter uses Costa's strategy with average power $P_{\mathrm{dpc}}$ and the knowledge of effective interference $\wtildemn{X}_0$, and with the DPC-parameter~\cite[Section II]{CostaDirtyPaper} $\alpha$ allowed to be arbitrary. For completeness, this strategy is outlined below:

Define random variable $T:=V_{\mathrm{dpc}}+\alpha \widetilde{X}_0$, where $V_{\mathrm{dpc}}\sim\mathcal{N}(0,P_{\mathrm{dpc}})$ and $\widetilde{X}_0\sim\mathcal{N}(0,\widetilde{\sigma}^2)$ are independent random variables. Similarly, define $Y:=V_{\mathrm{dpc}}+\widetilde{X}_0+Z$, with $Z\sim\mathcal{N}(0,1)$ is independent of $V_{\mathrm{dpc}}$ and $\widetilde{X}_0$. %We assume in the following that $\alpha, \beta$, and $P$ can be chosen, and are chosen, such that $R < I(T;Y)-I(\widetilde{X}_0;T)$.

The codebook $\mathcal{C}=\{\m{t}(i)\}_{i=1}^{2^{mR_T}}$ consists of $2^{mR_T}$ random codewords $\m{t}(i)$ (the ``auxiliary codewords''), where $R_T=I(T;Y)-\epsilon_1$ (for small $\epsilon_1>0$). Each element of each codeword is generated iid using the distribution of $T$. These codewords are then binned uniformly in $2^{mR}$ bins, each corresponding to a message $M\in\{1,2,\ldots,nR\}$. 

\textit{Encoding}: Given $\wtildemn{X}_0$ and the message $M$, the encoder chooses a sequence $\m{T}$ randomly from the set of sequences in the $M$-th bin that are jointly typical with $\wtildemn{X}_0$. The encoder then chooses $\m{V}_{\mathrm{dpc}}=\m{T}-\alpha\wtildemn{X}_0$. If $\wtildemn{X}_0,\m{T}$ are not jointly typical for any $\m{T}\in\mathcal{C}$, then the encoder declares an error and uses by default $\m{T}=\m{0}$, the all-zero codeword. 

\textit{Decoding}: The decoder looks for a unique codeword $\m{T}\in\mathcal{C}$ that is jointly typical with $\m{Y}$ (if such a unique sequence exists). If no such sequence exists, the decoder declares an error. 

\textit{Estimating $\m{X}_1$}: The decoder performs an MMSE estimation\footnote{We refer the reader to~\cite{WitsenhausenJournal} for the algebra of this joint estimation.} for estimating $\m{X}_1 =\m{X}_0+\m{V}=(1-\beta)\m{X}_0+\m{V}_{\mathrm{dpc}}$ using the channel output $\m{Y}=(1-\beta)\m{X}_0+\m{V}_{\mathrm{dpc}}+\m{Z}$ and the decoded codeword $\m{T}=\alpha(1-\beta)\m{X}_0 + \m{V}_{\mathrm{dpc}}$ (if $\m{T}$ has been decoded). 

%\textit{Claim}: The decoder succeeds in recovering $\m{T}$ correctly with probability converging to $1$ as $m\to\infty$. 

The total average power used by this scheme is asymptotically $P=P_{\mathrm{lin}}+P_{\mathrm{dpc}}$, because  $\m{V}_{\mathrm{dpc}}$ is asymptotically uncorrelated with $\wtildemn{X}_0$~\cite{CostaDirtyPaper} (because they are drawn independently), and thus also asymptotically uncorrelated with $\m{V}_{\mathrm{lin}}=-\beta\instvec$. 

While our strategy requires recovering the auxiliary codeword $\m{T}$, arguments in~\cite{CostaDirtyPaper,gelfandpinsker} only establish that the message $M$ is recovered with high probability. Can the chosen auxiliary codeword $\m{T}$ also be recovered? Techniques in~\cite{CostaDirtyPaper,gelfandpinsker} do not suffice for showing this\footnote{We thank the anonymous reviewer who made this observation.}. Here, the competing codewords in the same bin as the true codeword are not independent of the channel output, even under random coding. Thus it is not clear how to compute the the probability that one of them happens to be typicalÊwith the channel output. %This dependence between codewords in the bin of the transmitted codeword and the channel output is present in traditional DPC setup as well, but it does not matter there because the achievability requires recovery of any codeword that is j

The proof of the theorem below essentially shows that this dependence does not affect error probabilities significantly, and thus the chosen $\m{T}$-codeword can be recovered with high probability. The reason is that even though this encoding rule induces dependence between codewords, this dependence is not ``significant.''

%The way one usually bounds the error probability is% \footnote{For instance, if two codewords are extremely close to each other,  then if one is jointly-typical with the given $\m{x}_0$, so is very likely the other. Now, if they belong to the same message bin, they compete at the encoder to be chosen, reducing each-other's probability of being chosen. If they do not belong to the same message bin, }.

%a different message bins are independent of the chosen codeword, but because the codewords in the same message bin as the chosen codeword have to compete at the \textit{encoder}, this can induce a dependence between the codewords in the bin of the message. [***]

%The proofs in~\cite{CostaDirtyPaper,gelfandpinsker} depend crucially on independence of codewords that compete for joint-typicality with the channel output $\m{Y}$. [***]

%Because the auxiliary codeword $\m{T}$ at the transmitter is chosen among all codewords jointly typical with $\m{X}_0$ that \textit{lie in the same message-bin}, the encoding introduces dependence between $\m{T}$ codewords inside the bin, but the codewords outside the bin are still independent. Codewords outside the bin are still independent of the chosen codeword, and therefore of the channel output. Thus the usual joint-typicality concepts can be used. But here we ****

%The following theorem formally show that the chosen $\m{T}$-codeword can be recovered with high probability.
\begin{theorem}
\label{thm:achievable}
With probability converging to $1$ as $m\rightarrow\infty$, the chosen $\m{T}$ codeword is decoded by the joint-typicality test described above as long as the communication rate $R$ is bounded as follows:
\begin{equation}
\label{eq:alphabeta}
R (\alpha,\beta)< \frac{1}{2}\lo{\frac{P_{\mathrm{dpc}}(P_{\mathrm{dpc}}+\widetilde{\sigma}^2+1)}{P_{\mathrm{dpc}}\widetilde{\sigma}^2(1-\alpha)^2+ P_{\mathrm{dpc}}+\alpha^2\widetilde{\sigma}^2}}.
\end{equation}
\end{theorem}
\begin{proof}
The proof in Appendix~\ref{app:reliably} shows that the error probability converges to zero as long as \begin{equation}
R<I(T;Y)-I(X_0;T).
\label{eq:achrate}
\end{equation} 
That the expression in~\eqref{eq:achrate} equals that in~\eqref{eq:alphabeta} is shown in~\cite[Eq. (6)]{CostaDirtyPaper}.
\end{proof}
Keeping the rate $R$ constant, the obtained $\mathrm{MMSE}$ can now be minimized over the choice of $\alpha$ and $\beta$ under constraint~\eqref{eq:alphabeta}.

\subsubsection{Analytical upper bound on costs}
To obtain numerical upper bounds on the asymptotically optimal costs, we optimize~\eqref{eq:alphabeta} over $\alpha$ and $\beta$ for a given rate $R$ in Section~\ref{sec:dpcdes}. However, for analytical convenience in proving approximate-optimality of the family of strategies proposed here, we now derive looser upper bounds that have closed-form expressions. 

\begin{theorem}
For the problem as stated in Section~\ref{sec:probstat}, for communicating reliably at rate $R$, the asymptotic optimal cost is upper bounded by:
\begin{equation}
\lim_{m\to\infty}\avgcost_{\text{opt}} \leq \min\left\{k^2 2^{2R},  k^2 (2^{2R}-1) + \min\left\{1, \frac{\sigma^2}{2^{4R}+(2^{2R}-1)\sigma^2} \right\},     k^2 (\sigma^2+2^{2R}-1)  \right\}.
\end{equation}
\end{theorem}
\begin{proof}
See Appendix~\ref{app:upperbound}.
\end{proof}

%%%%%%%%%%%%%%%%%%%%%%%%%%%%%%%%%%%%%%%%%
\subsection{Tightness at $\mathrm{MMSE}=0$.}
\label{sec:match}
\begin{theorem}
\label{thm:match}
For the problem as stated in Section~\ref{sec:probstat}, for communicating reliably at rate $R$, let $R_0$ be defined as follows
\begin{equation}\label{eq:upRate}
R_0:=\sup_{\corr\in [-\sigma\sqrt{P},0]}\frac{1}{2}\lo{ \frac{   (P\sigma^2 - \corr^2) (1+\sigma^2+P+2\corr) }{\sigma^2(\sigma^2+P+2\corr)}}.
\end{equation}
If $R_0\geq 0$, a combination of linear and DPC-based strategies can (asymptotically) achieve $\mathrm{MMSE}=0$ and rate $R=R_0$, and no scheme can achieve $\mathrm{MMSE}=0$ with $R>R_0$. If $R_0<0$, no scheme can achieve $\mathrm{MMSE}= 0$. Further, inverting~\eqref{eq:upRate} provides the minimum required power $P$ to communicate at rate $R=R_0$. 
%
% (that is achievable by any scheme) in the perfect recovery limit $\mathrm{MMSE}(P,R)=0$, where $C(P)$ is given by
%\begin{equation}
%\label{eq:upRate}
%C(P)=\max\left\{\sup_{\corr\in [-\sigma\sqrt{P},0]}\frac{1}{2}\lo{ \frac{   (P\sigma^2 - \corr^2) (1+\sigma^2+P+2\corr) }{\sigma^2(\sigma^2+P+2\corr)} },0\right\}.
%\end{equation}
%That is, the minimum required power for achieving rate $R$ in the perfect-recovery limit is given by
%\begin{equation}
%\sup_{\corr\in [-\sigma\sqrt{P},0]}\frac{1}{2}\lo{ \frac{   (P\sigma^2 - \corr^2) (1+\sigma^2+P+2\corr) }{\sigma^2(\sigma^2+P+2\corr)}} = R.
%\end{equation}
%
\end{theorem}
\begin{proof}

%%%%
\textit{The achievability}

As shown in Theorem~\ref{thm:achievable}, the combination of linear and DPC-based strategies of~\cite{WitsenhausenJournal} recovers $\m{V}_{\mathrm{dpc}}+\alpha(1-\beta)\instvec$ at the decoder with high probability. In order to perfectly recover $\xonevec =\m{V}_{\mathrm{dpc}}+(1-\beta)\instvec$, we can use $\alpha=1$, and hence from Theorem~\ref{thm:achievable}, an achievable rate is:
\begin{equation}
R_{\mathrm{ach}} = \sup_{P_{\mathrm{lin}},P_{\mathrm{dpc}}: P=P_{\mathrm{lin}}+P_{\mathrm{dpc}}}\frac{1}{2}\lo{\frac{P_{\mathrm{dpc}}(P_{\mathrm{dpc}}+\widetilde{\sigma}^2+1)}{P_{\mathrm{dpc}}+\widetilde{\sigma}^2}},
\end{equation}
where we take a supremum over $P_{\mathrm{lin}},P_{\mathrm{dpc}}$ such that they sum up to $P$. Let $\corr= -\sigma\sqrt{P_{\mathrm{lin}}}$ (note that as $P_{\mathrm{lin}}$ varies from $0$ to $P$, $\corr$ varies from $0$ to $-\sigma\sqrt{P}$). Then, $P_{\mathrm{dpc}}= P - \frac{\corr^2}{\sigma^2}$, and $P_{\mathrm{dpc}}+\widetilde{\sigma}^2 = P_{\mathrm{dpc}}+\sigma^2 + P_{\mathrm{lin}}- 2\sigma\sqrt{P_{\mathrm{lin}}} = P+\sigma^2 +2\corr$. 
% Accommodating for the possibility that the supremum could be negative,
%\begin{eqnarray}
%\label{eq:lowRate}
%R_{\mathrm{ach}} = \max\left\{\sup_{\corr\in [-\sigma\sqrt{P},0]}\frac{1}{2}\lo{\frac{\left(P - \frac{\corr^2}{\sigma^2}\right)( P+\sigma^2 +2\corr+1)}{ P+\sigma^2 +2\corr}},0\right\}.
%\end{eqnarray}
Simple algebra shows that this expression matches that in Theorem~\ref{thm:match}. %The power below which the expression becomes zero is the minimum required power to communicate at rate $R$ using this scheme. 

\vspace{0.1in}

%%%%
\textit{The converse}

%Among the steps in the proof of Theorem~\ref{thm:newlower}, at least two steps employ potentially loose bounding techniques --- step $(a)$ in~\eqref{eq:lossy} which removes conditioning to upper bound a differential entropy term, and~\eqref{eq:align}, which assumes that  $X-\widehat{X}$ and $(1-\gamma)S - \gamma U$ can be perfectly aligned. 

%Observe that if $\mathrm{MMSE}=0$, $\widehat{X}=X$ almost-surely. Now, in step $(a)$ of~\eqref{eq:lossy}, if $\gamma$ is chosen so that $\gamma X$($=\gamma\widehat{X}$ a.s) is an MMSE estimate of $S$, then $S-\gamma X independent \gamma X$, because $S$ and $X$ are jointly-Gaussian. Thus there is no loss in removing the conditioning for $\mathrm{MMSE}=0$ for $\gamma$ chosen as detailed. Also, because $X-\widehat{X}=0$ almost-surely (in the limit), there is no loss in assuming that $X-\widehat{X}$ can be aligned with $(1-\gamma)S - \gamma U$ either. This offers hope that at $\mathrm{MMSE}=0$, the bound may indeed be tight. The following arguments show that this is indeed true.

Let $P$ be large enough\footnote{As we shall see, this minimum required power at zero rate for perfect reconstruction can be obtained by inverting~\eqref{eq:upRate}.} so that perfect reconstruction of $\m{X}_1$ is possible at rate 0. % (communication at a non-zero rate will require larger power). 
 Since we are free to choose $\gamma$ in Theorem~\ref{thm:newlower}, let $\gamma = \gamma^* = \frac{\sigma^2+ \corr}{\sigma^2+P+2\corr}$. Then, $1-\gamma^* = \frac{P+ \corr}{\sigma^2+P+2\corr}$. Thus, we get
\begin{equation}
0\geq \inf_{\corr} \frac{1}{\gamma^{*^2}}\left(\left( \sqrt{\frac{\sigma^2 2^{2R}}{1+\sigma^2+P+2\corr}}  -   \sqrt{(1-\gamma^*)^2\sigma^2 + \gamma^{*^2} P -2\gamma^* (1-\gamma^*)\corr}     \right)^+\right)^2.
\end{equation}
It has to be the case that the term inside $(\cdot{})^+$ is non-positive for some value of $\corr$. This immediately yields\footnote{The algebra below is included for the convenience of the reviewers and can be removed in the final version.}
\begin{eqnarray*}
2^{2R}&\leq &\sup_{\corr} \frac{1}{\sigma^2}\left((1-\gamma^*)^2\sigma^2 + \gamma^{*^2} P -2\gamma^* (1-\gamma^*)\corr \right) (1+\sigma^2+P+2\corr)\\
& = & \textcolor{red}{\sup_{\corr}\frac{1}{\sigma^2}\frac{\left( (P+\corr)^2\sigma^2 + (\sigma^2 + \corr)^2 P - 2 (P+\corr)(\sigma^2+\corr)\corr  \right)}{(\sigma^2+P+2\corr)^2}(1+\sigma^2+P+2\corr)}\\
& =&\textcolor{red}{\sup_{\corr} \frac{1}{\sigma^2}\frac{\left(   P^2\sigma^2 - \corr^2\sigma^2 +2P\corr\sigma^2 + P\sigma^4 - P\corr^2 -2\corr^3  \right)}{(\sigma^2+P+2\corr)^2} (1+\sigma^2+P+2\corr)}\\
& =&\textcolor{red}{\sup_{\corr} \frac{1}{\sigma^2}\frac{\left(   (P\sigma^2 - \corr^2) (P+\sigma^2 + 2\corr)  \right)}{(\sigma^2+P+2\corr)^2} (1+\sigma^2+P+2\corr)}\\
& =&\sup_{\corr} \frac{   (P\sigma^2 - \corr^2) (1+\sigma^2+P+2\corr) }{\sigma^2(\sigma^2+P+2\corr)} 
\end{eqnarray*}
Thus, we get the following upper bound on achievable rate $R$ with perfect reconstruction of $\m{X}_1$,
\begin{equation}
\label{eq:upRate2}
R\leq \sup_{\corr\in [-\sigma\sqrt{P},\sigma\sqrt{P}]}\frac{1}{2}\lo{ \frac{   (P\sigma^2 - \corr^2) (1+\sigma^2+P+2\corr) }{\sigma^2(\sigma^2+P+2\corr)} }.
\end{equation}
The term $(P\sigma^2-\corr^2)$ does not depend on the sign of $\corr$. However, the  term 
\begin{equation}
\frac{1+\sigma^2+P+2\corr}{\sigma^2+P+2\corr} =1 +  \frac{1}{\sigma^2+P+2\corr}
\end{equation}
is clearly larger for $\corr<0$ if we fix $|\corr|$. Thus the supremum in~\eqref{eq:upRate2} is attained at some $\corr<0$, and we get
\begin{equation}
\label{eq:upRate3}
R\leq \sup_{\corr\in [-\sigma\sqrt{P},0]}\frac{1}{2}\lo{ \frac{   (P\sigma^2 - \corr^2) (1+\sigma^2+P+2\corr) }{\sigma^2(\sigma^2+P+2\corr)} }=R_0,
\end{equation}
from~\eqref{eq:upRate}. Thus for perfect reconstruction ($\mathrm{MMSE}=0$), the  combination of linear and DPC strategies proposed in~\cite{WitsenhausenJournal} is optimal.
\end{proof}

\subsection{Approximate optimality over all problem parameters}
\begin{theorem}
\label{thm:approximate}
Let $\mathrm{MMSE}_{lb}(P)$ denote the lower bound from Theorem~\ref{thm:newlower}, that is,
 \begin{equation}
\mathrm{MMSE}_{lb}(P):=\inf_{\corr}\sup_{\gamma\in\mathbb{R}} \frac{1}{\gamma^2}\left(\left( \sqrt{\frac{\sigma^2 2^{2R}}{1+\sigma^2+P+2\corr}}  -   \sqrt{(1-\gamma)^2\sigma^2 + \gamma^2 P -2\gamma (1-\gamma)\corr}     \right)^+\right)^2.
\end{equation}
Then, the following bounds characterize the optimal cost for the problem stated in Section~\ref{sec:probstat} in the asymptotic limit of $m\rightarrow\infty$:
\begin{equation}
\label{eq:approx}
\inf_{P\geq 2^{2R}-1} k^2 P + \mathrm{MMSE}_{lb}(P) \leq \avgcost_{\text{opt}}\leq 16\left( \inf_{P\geq 2^{2R}-1} k^2 P + \mathrm{MMSE}_{lb}(P)  \right),
\end{equation} 
uniformly over all rates $R$ and all parameters $k,\sigma^2$.
\end{theorem}
\begin{proof}
See Appendix~\ref{app:approx}.
\end{proof}

\section{Numerical results}

\begin{figure}[htb]
\begin{center}
\includegraphics[width=3.4in]{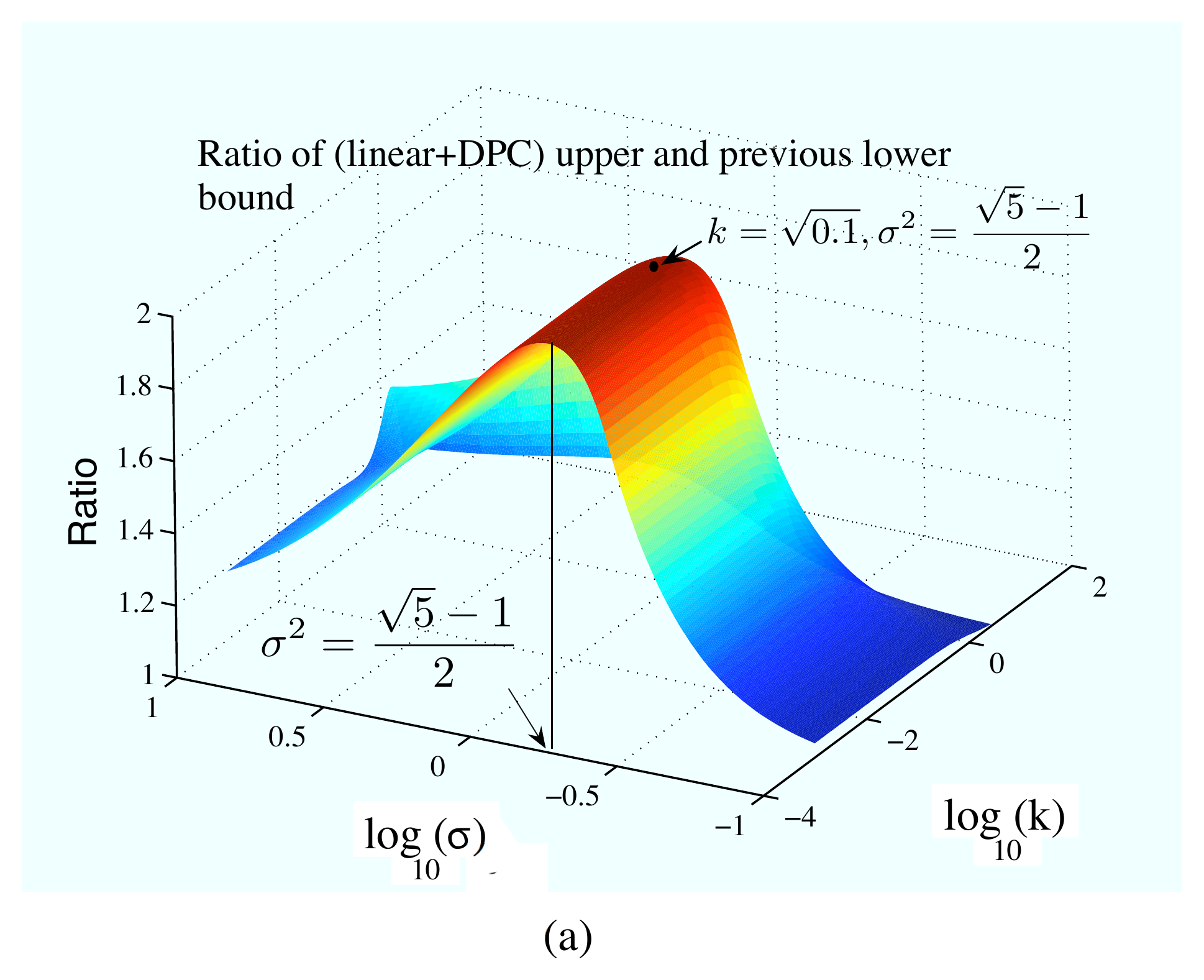}\includegraphics[width=3.4in]{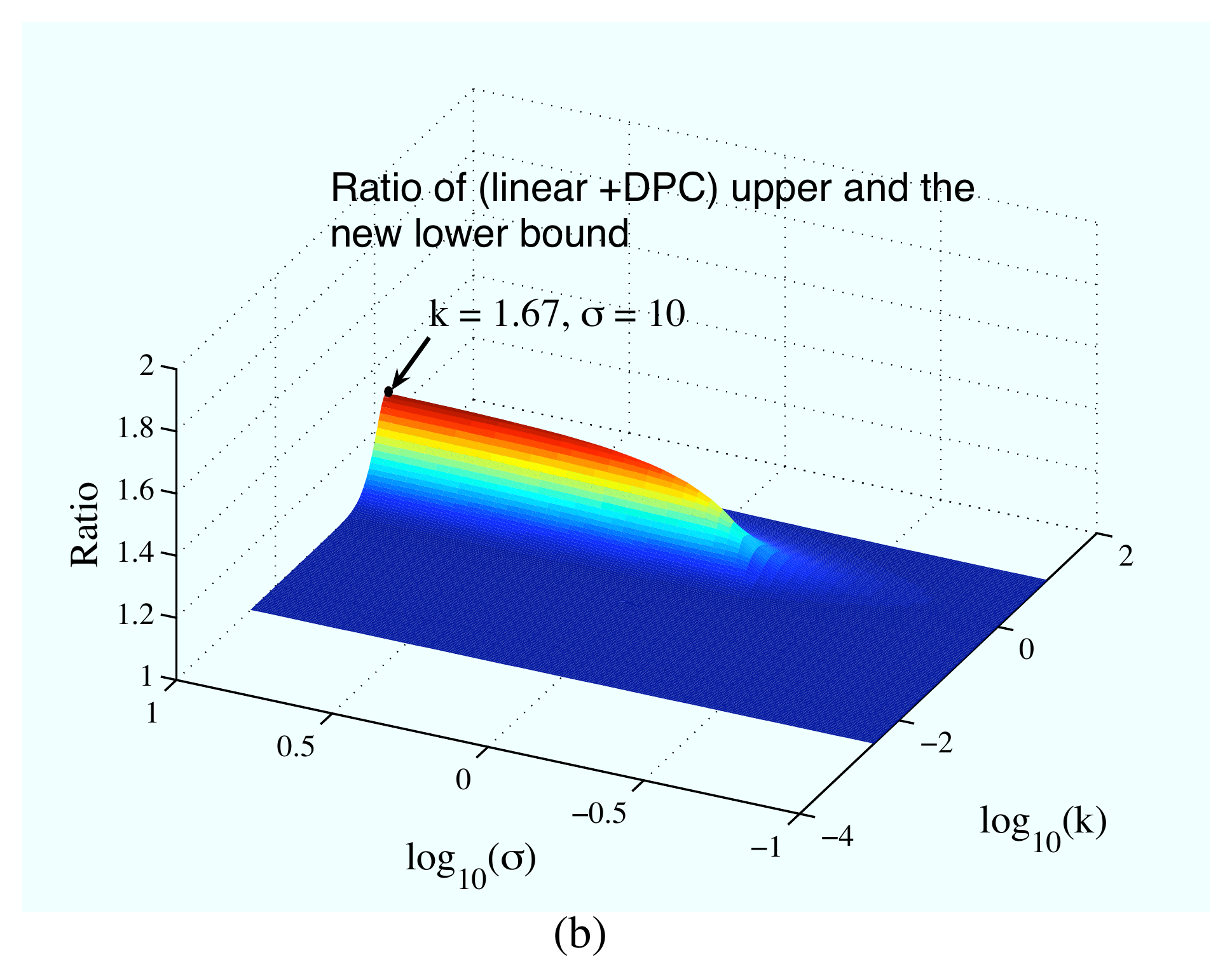}
\caption{The ratio of upper and lower bounds on the total asymptotic cost for the vector Witsenhausen counterexample (\textit{i.e.}, $R=0$ case) with the lower bound taken from~\cite{WitsenhausenJournal} in (a) and from Corollary~\ref{coro:wit} in (b). As compared to the previous best known ratio of $2$~\cite{WitsenhausenJournal}, the ratio here is smaller than $1.3$. Further, an infinitely long ridge along $\sigma^2=\frac{\sqrt{5}-1}{2}$ and small $k$ that is present in lower bounds of~\cite{WitsenhausenJournal} is no longer present here. This is a consequence of the tightness lower bound at $\mathrm{MMSE}=0$, and hence for small $k$. A ridge remains along $k\approx 1.67$ ($\log_{10}(k)\approx 0.22$) and large $\sigma$, and this can be understood by observing Fig.~\ref{fig:sigma} for $\sigma=10$. }
\label{fig:ratio}
\end{center}
\end{figure}

\begin{figure}[htb]
\begin{center}
\includegraphics[width=6.5in]{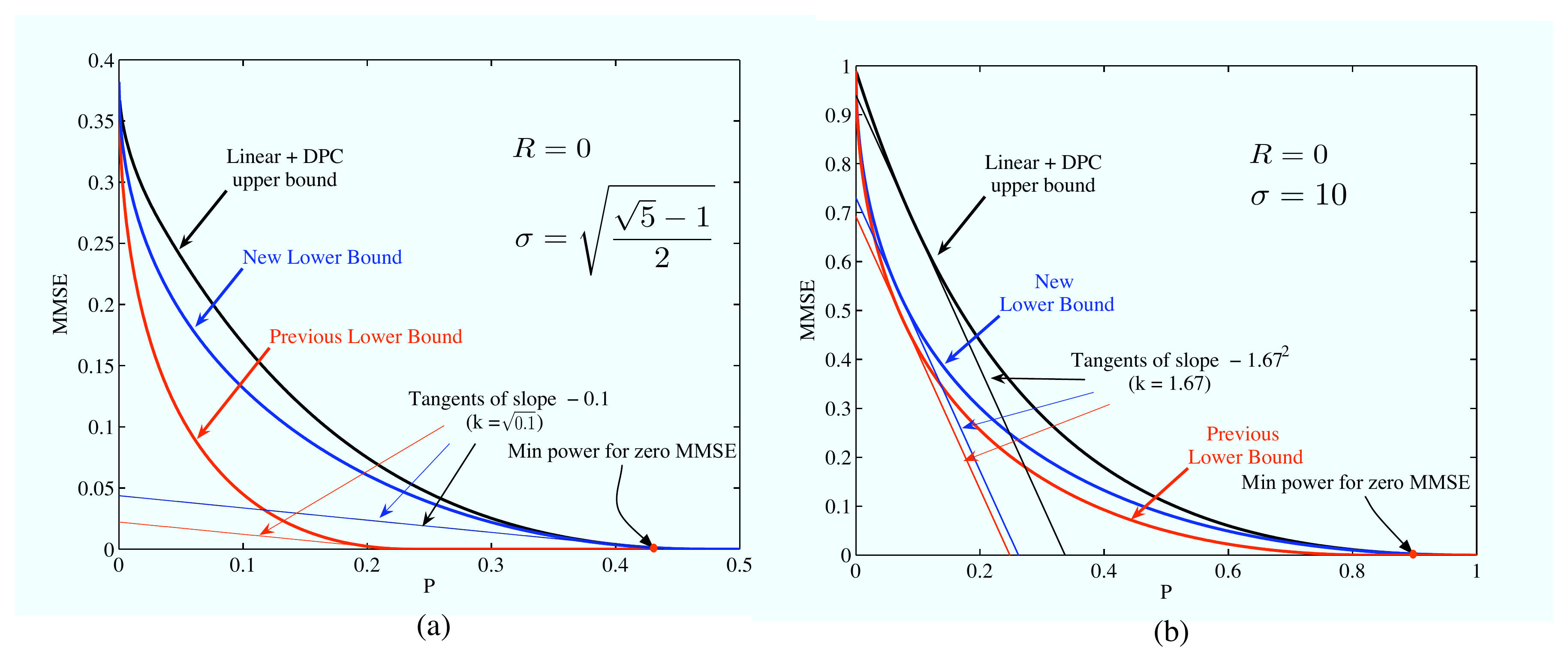}
\caption{Upper and lower bounds on  asymptotic $\mathrm{MMSE}$ vs $P$ for $\sigma=\sqrt{\frac{\sqrt{5}-1}{2}}$ (square-root of the Golden ratio; Fig. (a)) and $\sigma=10$ (b) for zero-rate (the vector Witsenhausen counterexample). The bounds match at $\mathrm{MMSE}=0$, characterizing the minimum power for perfect reconstruction of $\m{X}_1$. Tangents are drawn to evaluate the total cost for $k=\sqrt{0.1}$ for $\sigma=\sqrt{\frac{\sqrt{5}-1}{2}}$, and for  $k=1.67$ for $\sigma=10$ (slope $= -k^2$). The intercept on the $\mathrm{MMSE}$ axis of the tangent provides the respective bound on the total cost. The tangents to the upper bound and the new lower bound almost coincide for small values of $k$. At $k\approx 1.67$ and $\sigma=10$, however, our bound is not  significantly better than that in~\cite{WitsenhausenJournal} and hence the ridge along $k\approx 1.67$ remains in the new ratio plot in Fig.~\ref{fig:ratio}.}
\label{fig:sigma}
\end{center}
\end{figure}

%\begin{figure}[htb]
%\begin{center}
%\includegraphics[scale=0.47]{SigmaSqGoldenRatio2}\\\includegraphics[scale=0.47]{Sigma10_3}
%\caption{Upper and lower bounds on $P$ vs asymptotic $\mathrm{MMSE}$ for $\sigma=\sqrt{\frac{\sqrt{5}-1}{2}}$ (square-root of the Golden ratio; Fig. (a)) and $\sigma=10$ (b) for zero-rate (the vector Witsenhausen counterexample). Tangents are drawn to evaluate the total cost for $k=\sqrt{0.1}$ for $\sigma=\sqrt{\frac{\sqrt{5}-1}{2}}$, and for  $k=1.67$ for $\sigma=10$ (slope $= -k^2$). The intercept on the MMSE axis of the tangent provides the respective bound on the total cost. The tangents to the upper bound and the new lower bound almost coincide for small values of $k$. At $k\approx 1.67$ and $\sigma=10$, however, our bound is not  significantly better than that in~\cite{WitsenhausenJournal} and hence the ridge along $k\approx 1.67$ remains in the new ratio plot in Fig.~\ref{fig:ratio}.}
%\label{fig:sigma}
%\end{center}
%\end{figure}

\begin{figure}[htb]
\begin{center}
\includegraphics[scale=0.35]{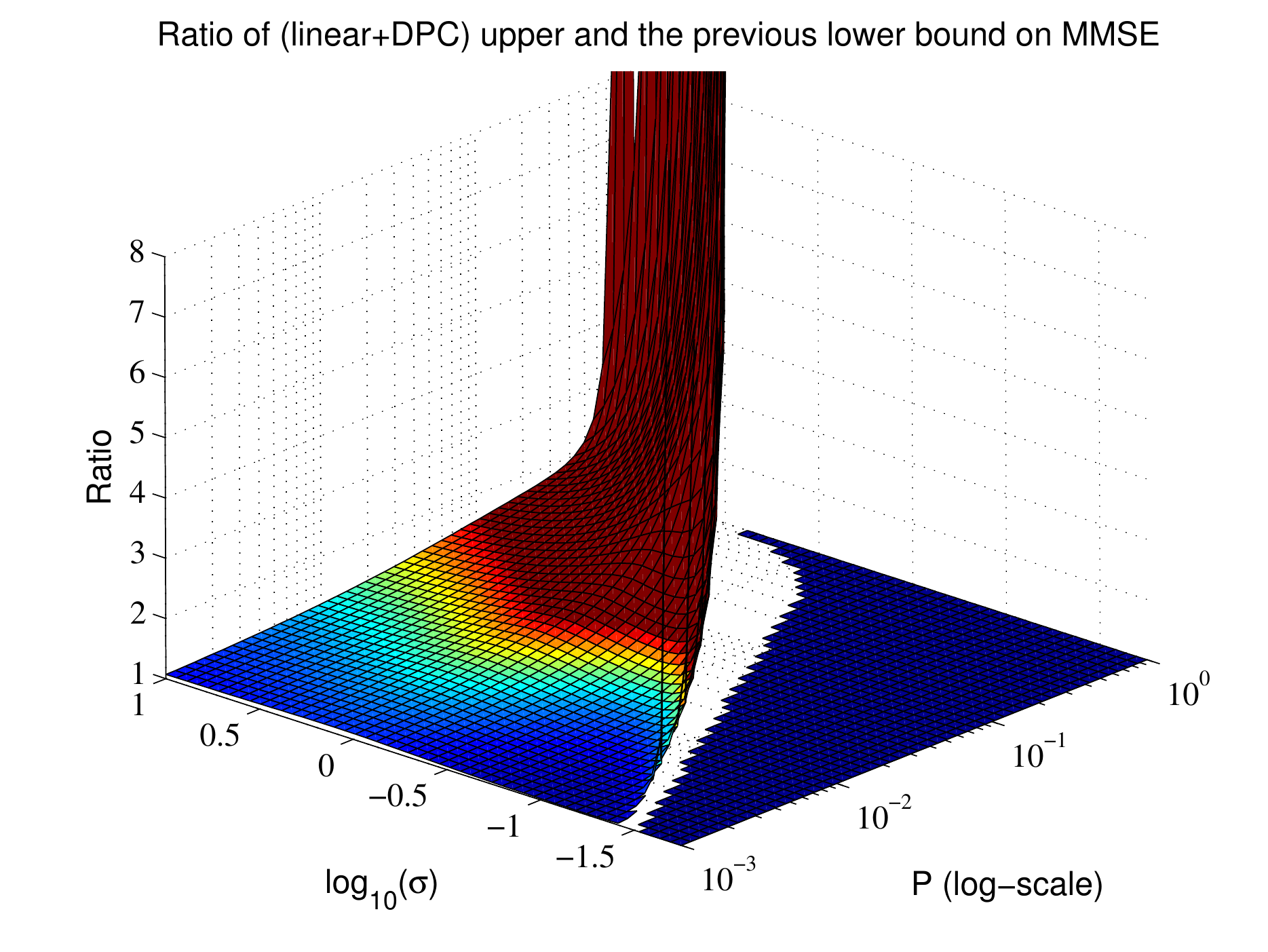}
\includegraphics[scale=0.35]{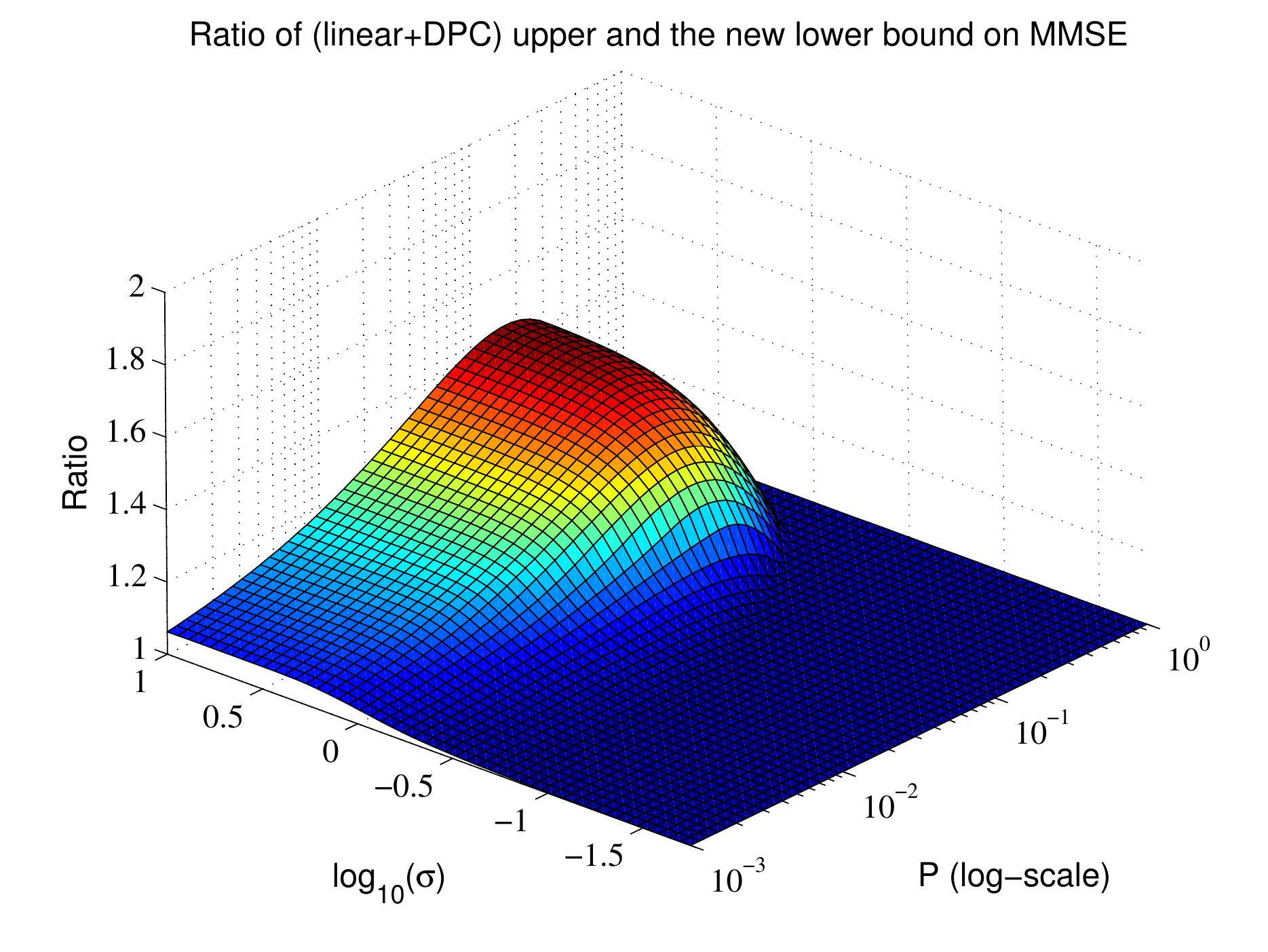}
\caption{Ratio of upper and lower bounds on $\mathrm{MMSE}$ vs $P$ and $\sigma$ at $R=0$. Whereas the ratio diverges to infinity with the old lower bound of~\cite{WitsenhausenJournal} (left), it is bounded by $1.5$ for the new bound (right) because of the improved tightness of the new bound at small $\mathrm{MMSE}$.}
\label{fig:MMSEratio}
\end{center}
\end{figure}

\begin{figure}[htb]
\begin{center}
\includegraphics[scale=0.35]{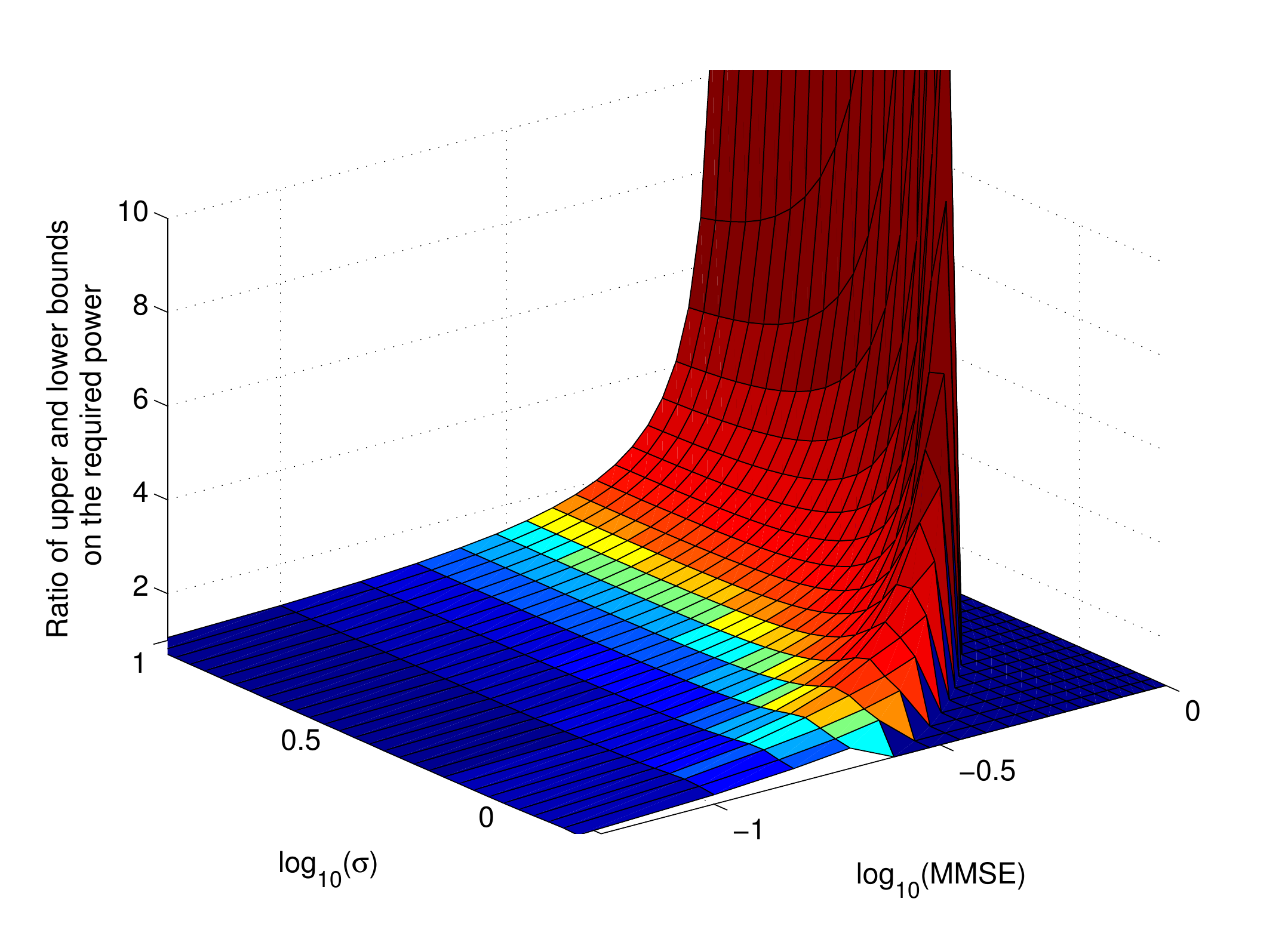}
\caption{Ratio of upper and lower bounds on $P$ vs $\mathrm{MMSE}$ and $\sigma$ at $R=0$. Interestingly, the ratio increases to infinity as $\sigma\rightarrow\infty$ along the path where $P$ is close to zero (corresponding to ``high" $\mathrm{MMSE}=\frac{\sigma^2}{\sigma^2+1}$).}
\label{fig:PowerRatio}
\end{center}
\end{figure}

%
%\begin{figure}[htb]
%\begin{center}
%\includegraphics[scale=0.55]{RatePoint5GoldenRatio}
%\caption{Upper and lower bounds on $P$ vs $\mathrm{MMSE}$ for $\sigma=\sqrt{\frac{\sqrt{5}-1}{2}}$\ for $R=0.5$. Though the bounds match at $\mathrm{MMSE}=0$ (by Theorem~\ref{thm:match}), the bounds do not match at the minimum power ($P=1$ here) for nonzero rates. Below $P=1$, communication at $R=0.5$ is not possible.}
%\label{fig:nonzerorate}
%\end{center}
%\end{figure}
%
%
%
%\begin{figure}[htb]
%\begin{center}
%\includegraphics[scale=0.53]{RvsMMSEPower1}
%\caption{Plot of upper and lower bounds on $\mathrm{MMSE}$ vs rate for fixed power $P=1$ and $\sigma=\sqrt{\frac{\sqrt{5}-1}{2}}$. Higher rates require higher average distortion in the reconstruction of $\xonevec $.}
%\label{fig:RvsMMSE}
%\end{center}
%\end{figure}

In Theorem~\ref{thm:approximate}, we showed that the ratio of upper and lower bounds on asymptotic costs is bounded by $16$ for all rates and all problem parameters. For the special case of rate $0$ (which corresponds to the Witsenhausen counterexample),~\cite[Theorem 1]{WitsenhausenJournal} proves that the ratio of the upper bound used here, and the looser lower bound of~\cite{WitsenhausenJournal} is bounded by $11$. Further, numerical calculations in~\cite{WitsenhausenJournal} (reproduced here in Fig.~\ref{fig:ratio}) show that the ratio is bounded by $2$. Using the improved lower bound obtained in Theorem~\ref{thm:newlower}, numerical calculations plotted in Fig.~\ref{fig:ratio}(b) show that
in the asymptotic limit $m\rightarrow\infty$, the ratio of the upper and the new lower bounds (from Corollary~\ref{coro:wit}) on the weighted cost is bounded by $1.3$, an improvement over the ratio of $2$
in~\cite{WitsenhausenJournal}. 

The ridge of ratio $2$
along $\sigma^2=\frac{\sqrt{5}-1}{2}$ present in
Fig.~\ref{fig:ratio}(a) (obtained using the old bound from~\cite{WitsenhausenJournal}) does not exist with the new lower
bound since this small-$k$ regime corresponds to target $\mathrm{MMSE}$s close
to zero -- where the new lower bound is tight. This is illustrated in
Fig.~\ref{fig:sigma} (a). Also shown in Fig.~\ref{fig:sigma} (b) is the lack of tightness in the bounds at small $P$. The figure explains how this looseness results in the ridge along $k\approx 1.67$ still surviving in the new ratio plot.

Fig.~\ref{fig:MMSEratio} shows the ratio of upper and lower bounds on $\mathrm{MMSE}(P,0)$ versus $P$ and $\sigma$. While the ratio with the bound of~\cite{WitsenhausenJournal} was unbounded (Fig.~\ref{fig:MMSEratio}, left), the new ratio is bounded by a factor of $1.5$ (Fig.~\ref{fig:MMSEratio}, right). This is again a reflection of the tightness of the bound at small $\mathrm{MMSE}$. A flipped perspective is shown in Fig.~\ref{fig:PowerRatio}, where we compute the ratio of upper and lower bounds on required power to attain a specified $\mathrm{MMSE}$. As further evidence of the lack of tightness in the small-$P$ (``high distortion'') regime, the ratio of upper and lower bounds on required power diverges to infinity along the path $\mathrm{MMSE}=\frac{\sigma^2}{\sigma^2+1}$. This indicates that while we have a good understanding of good strategies close to zero reconstruction error, we have little understanding of strategies that help us get good returns on power investment close to zero power. 

%Fig.~\ref{fig:nonzerorate} shows the upper and the lower bounds for $R=0.5$. Again, the bounds are not tight in the small-$P$ regime: now the looseness is at the lowest power $P=1$ at which communication at $R=0.5$ is possible. As shown in Theorem~\ref{thm:match}, the bounds are still tight at $\mathrm{MMSE}=0$. Fig.~\ref{fig:RvsMMSE} shows the upper and lower bounds on $\mathrm{MMSE}$ as a function of the rate $R$ for fixed power $P=1$ and $\sigma^2$ equal to the Golden ratio. The figure demonstrates that beyond the maximum rate with zero distortion, the price of increasing rate is an increased distortion in the estimation of $\xonevec $. 

The MATLAB code for these figures can be found in~\cite{CodeForInformationEmbedding}.

\section*{Acknowledgments}
P.~Grover and A.~Sahai acknowledge the support of the National Science Foundation (CNS-403427, CNS-093240, CCF-0917212 and CCF-729122) and Sumitomo Electric. A.~B.~Wagner acknowledges the support of NSF CSF-06-42925 (CAREER) grant and AFOSR grant FA9550-08-1-0060. We thank Hari Palaiyanur, Se Yong Park, Gireeja Ranade and Pravin Varaiya for helpful discussions, and the anonymous reviewers for valuable suggestions that improved this paper significantly.

\appendices{}

%%%%%%%%%%%%%%%%%%%%%%%%%%%%%%%
\section{Codewords $\m{T}$ can be recovered reliably}
\label{app:reliably}
\begin{definition}
A vector $\m{y}$ is said to be \textit{typical} (w.r.t. to a distribution $f_Y(y)$) if:
\begin{eqnarray*}
\left|\frac{1}{m}\sum_i\log f_Y(Y_i) + h(Y)  \right|\leq\epsilon.
 \end{eqnarray*}
\end{definition}

\begin{definition}
Vectors $\m{x}$ and $\m{y}$ are said to be \textit{jointly-typical} with respect to the joint distribution $f_{X,Y}(x,y)$ if $\m{x}$ is typical w.r.t $f_X(x)$, $\m{y}$ is typical w.r.t. $f_Y(y)$, and 
\begin{eqnarray*}
 \left|\frac{1}{m}\sum_i\log f_{X,Y}(X_i,Y_i) + h(X,Y)  \right|\leq\epsilon.
 \end{eqnarray*}
\end{definition}
\begin{definition}
The set of vectors $\m{x}$ jointly typical with $\m{y}$, denoted by $\mathcal{T}_{\m{y}}$, is defined as
\begin{eqnarray*}
\mathcal{T}_{\m{y}} : = \left\{   \begin{array}c
 \mathrm{Empty\; set} \;\;\;\text{if $\m{y}$ is not typical}   \\
\{ \m{x}: (\m{x},\m{y})\; \text{are jointly-typical} \}\;\;\text{if $\m{y}$ is typical}    \end{array}\right.
\end{eqnarray*}
\end{definition}
%%%%%%%%%%%%%%%%%%%%
\begin{lemma}
\label{lem:type}
Given a joint distribution $f_{X,Y}(x,y)$ for $X$ and $Y$, for a typical vector $\m{y}$ (w.r.t. $f_Y(y)$), $\mathrm{Vol}(\mathcal{T}_{\m{y}}):= \int_{\m{x}\in\mathcal{T}_{\m{y}}} d\m{x}\leq 2^{m(h(X|Y)+2\epsilon)}$, and $\Pr(\m{X}\in\mathcal{T}_{\m{y}})\leq 2^{-m(I(X;Y)-3\epsilon)}$.
\end{lemma}
%%%%
\begin{proof}
\begin{eqnarray}
\nonumber\Pr\left(  \m{X}\in\mathcal{T}_{\m{y}}   \right) &=& \int_{\m{x}\in\mathcal{T}_{\m{y}} }f(\m{x})d\m{x}\leq \int_{\m{x}\in\mathcal{T}_{\m{y}} } 2^{-m(h(X)-\epsilon)} d\m{x}\\
&=& 2^{-m(h(X)-\epsilon)} \int_{\m{x}\in\mathcal{T}_{\m{y}} } d\m{x}.
\label{eq:firstbound}
\end{eqnarray}
Now,
\begin{eqnarray*}
1&=&\int_{\m{x}\in\mathbb{R}^m} f(\m{x}|\m{y}) d\m{x} \geq \int_{\m{x}\in\mathcal{T}_{\m{y}}} f(\m{x}|\m{y}) d\m{x}=\int_{\m{x}\in\mathcal{T}_{\m{y}}} \frac{f(\m{x},\m{y})}{f(\m{y})} d\m{x}\\
&\geq & \int_{\m{x}\in\mathcal{T}_{\m{y}}} \frac{2^{-m(h(X,Y)+\epsilon)}}{2^{-m(h(Y)-\epsilon)}}d\m{x}= 2^{-m(h(X|Y)+2\epsilon)}\int_{\m{x}\in\mathcal{T}_{\m{y}}} d\m{x}=2^{-m(h(X|Y)+2\epsilon)}\mathrm{Vol}(\mathcal{T}_{\m{y}}).
\end{eqnarray*}
Thus, 
\begin{equation}
\label{eq:secondbound}
\mathrm{Vol}(\mathcal{T}_{\m{y}})\leq 2^{m(h(X|Y)+2\epsilon)}.
\end{equation}
Along with~\eqref{eq:firstbound},~\eqref{eq:secondbound} yields
\begin{equation}
\Pr\left(  \m{X}\in\mathcal{T}_{\m{y}}   \right)\leq 2^{-m(h(X)-\epsilon)}2^{m(h(X|Y)+2\epsilon)} = 2^{-m(I(X;Y)-3\epsilon)}
\end{equation}
\end{proof}
%%%%%%%%%%%
To derive the achievable rates for the DPC-based strategy of Section~\ref{sec:dpc}, we show below that as long as $R<I(T;Y)-I(X_0;T)$, the probability of error of recovering $\m{T}$ (and not just the message $M$) can be driven to zero. For simplicity, we assume that the scaling factor $\beta=0$, and thus $\wtildemn{X}_0=\m{X}_0$ and $P=P_{\mathrm{dpc}}$. The proof trivially extends to the case $\beta\neq 0$. Choose the rate $R_T$ as $R_T=I(T;Y)-\epsilon_1=I(X_0;T)+R+\epsilon_2$ for  $\epsilon_1,\epsilon_2>0$ satisfying $\epsilon_1 = 10\epsilon + 2\epsilon_2$ for some $\epsilon>0$. That such $\epsilon,\epsilon_1,$ and $\epsilon_2$ exist is true because $R<I(T;Y)-I(X_0;T)$. 

Let the codebook $\mathcal{C}=\{\m{t}(i)\}_{i=1}^{2^{mR_T}}$, where $\m{t}(i)$ are the codewords. Let the bins in which these codewords are uniformly distributed be denoted by $\mathcal{B}(M)$, where $M=1,2,\ldots,2^{mR}$ is the message-index. The ties in encoding are broken by choosing any $\m{t}$-codeword randomly from the set of codewords that lie in the message-bin and are jointly typical with $\m{X}_0$. The random variable $\Delta$ is used to denote this randomness in binning and breaking ties.

\begin{lemma}\label{lem:tdist}
For any codebook $\mathcal{C}=\{\m{t}(i)\}_{i=1}^{2^{mR_T}}$ with codewords $\m{t}(i)$ encoded according to the joint-typicality criterion above, the probability of any nonzero codeword $\m{t}(i)$ (averaged over realizations of $\m{X}_0$, $\Delta$ and $M$) being chosen is bounded by
\begin{equation}
\Pr(\m{t}(i)\;\text{chosen}|\mathcal{C}) \leq 2^{-m(R_T-3\epsilon-\epsilon_2)},
\end{equation}
where $R_T=I(X_0;T)+R+\epsilon_2$.
\end{lemma}
\begin{proof}
For given $\m{x}_0$, the encoder chooses randomly among the $\m{t}$-codewords in message-bin $M$ that are jointly typical with $\m{x}_0$. Thus, if $\m{t}(i)$ is not typical, then $\Pr(\m{t}(i)\;\text{chosen}|\mathcal{C})=0$ and the above bound trivially holds. If $\m{t}(i)$ is typical,  
\begin{eqnarray}
\nonumber\Pr(\m{t}(i)\;\text{chosen}|\mathcal{C}) &\leq& \Pr(\m{t}(i)\in\mathcal{B}(M)\cap\mathcal{T}_{\m{x}_0}) = \Pr(\{\m{t}(i)\in\mathcal{B}(M)\}\cap\{\m{t}(i)\in\mathcal{T}_{\m{x}_0}\})\\
\nonumber &=&\Pr(\{\m{t}(i)\in\mathcal{B}(M)\})\Pr(\{\m{t}(i)\in\mathcal{T}_{\m{x}_0}\})\\
&\leq& \frac{1}{2^{mR}}2^{-m(I(T;X_0)-3\epsilon)}=2^{-m(I(X_0;T)+R-3\epsilon)}=2^{-m(R_T-3\epsilon-\epsilon_2)}.
%\int_{\m{x}_0 \in\mathcal{T}_{\m{t}(i)}} f_{\m{X}_0}(\m{x}_0)d\m{x}_0 \times \Pr(\m{t}(i)\in\mathcal{B}(M))\\
%\nonumber&\overset{(a)}\leq & \int_{\m{x}_0 \in\mathcal{T}_{\m{t}(i)}} 2^{-m(h(X_0)-\epsilon)}d\m{x}_0 \times2^{-mR} = 2^{-m(h(X_0)+R-\epsilon)} \mathrm{Vol}(\mathcal{T}_{\m{t}(i)})\\
%&\overset{(b)}\leq & 2^{m(h(X_0|T)+2\epsilon)} 2^{-m(h(X_0)+R-\epsilon)}=2^{-m(I(X_0;T)+R-3\epsilon)}=2^{-m(R_T-3\epsilon-\epsilon_2)},
\label{eq:tms}
\end{eqnarray}
%where $(a)$ and $(b)$ are satisfied for any chosen $\epsilon>0$ for $m$ large enough; $(a)$ follows from the fact that $\m{X}_0$ is typical (a requirement for joint typicality), and $(b)$ follows from Lemma~\ref{lem:type}.
% (Theorem 9.2.2 in Cover and Thomas).
%
\end{proof}
We remark that the bound above holds for probability of sending non-$\m{0}$ codewords. As we will see, in our proof of achievability, the encoding-failure event (which is when $\m{0}$ is transmitted) does not matter because it has low probability. 

%\begin{lemma}
%\label{lem:indep}
%For a random variable $X$ and a random set $A$ drawn independently of $X$,
%\begin{equation}
%\Pr(X\in A) = \int_a  \Pr(X\in a) dF(a), 
%\end{equation} 
%where $dF(a)$ is the measure of set $a$ under the distribution of $A$. 
%\end{lemma}
%\begin{proof}
%\begin{eqnarray*}
%\Pr(X\in A) &=& \int_a \Pr(X\in a|A=a) dF(a)\\
%&=& \int_a\Pr(X\in a) dF(a),
%\end{eqnarray*}
%where the last equality follows from the fact that $X$ is independent of $A$.
%\end{proof}

Define the error-events $\mathcal{E}_1$, $\mathcal{E}_2$ and $\mathcal{E}_3$ as follows: \\$\mathcal{E}_1:= \{\text{no $\m{T}(j)$ is jointly typical with $\m{X}_0$}\}$, the encoding error-event; \\$\mathcal{E}_2:=\{\exists\; \whatmn{T} \in\mathcal{C}\cap\mathcal{T}_{\m{Y}} \;\text{s.t.}\whatmn{T}\neq \m{T}\}$; and \\$\mathcal{E}_3:=\{\text{$\m{T}$ and $\m{Y}$ are not jointly typical}\}$. Note that $\mathcal{E}_2\cup\mathcal{E}_3$ is the decoding error event.

The novelty in the proof lies in dealing with event $\mathcal{E}_2$ because the events $\mathcal{E}_1$ and $\mathcal{E}_3$ are the same as those used in proving the achievability for the Gelfand-Pinsker strategy~\cite[Pg. 180-181]{ElGamalKim}.\footnote{That $\Pr(\mathcal{E}_1)\to 0$ is a straightforward consequence of the fact that $R_T=I(X_0;T)+R+\epsilon_2$ and that $\mathcal{C}$ is a random codebook in the traditional sense. The proof can be found in~\cite[Pg. 180-181]{ElGamalKim}. For $\Pr(\mathcal{E}_3)$, we again refer the reader to~\cite[Pg. 180-181]{ElGamalKim} where it is shown that the chosen auxiliary codeword $\m{t}$ is jointly-typical with the channel output $\m{y}$ with high probability (although the result is shown for DMCs, it is clear that the result extends to continuous channels as well). } Focusing on the event $\mathcal{E}_2$, our requirement of reconstructing the $\m{t}$-codeword (and not just its message bin) imposes the condition that \textit{no other codeword} in $\mathcal{C}$ (and not merely the codewords outside the correct message bin) be jointly typical with the channel output $\m{y}$ (other than the chosen $\m{t}$-codeword). This introduces a certain dependence between the $\m{t}$-codewords under consideration (as discussed in Section~\ref{sec:dpcdes}), and our analysis of $\Pr(\mathcal{E}_2)$ below explicitly bounds this dependence.
%
%
%Now,
\begin{eqnarray}
\nonumber && \mathrm{Pr}_{\m{X}_0,\mathcal{C},\Delta,\m{Z}}(\mathcal{E}_2)= \mathrm{Pr}_{\m{X}_0,\mathcal{C},\Delta,\m{Z}}(\{\exists\; \whatmn{T} \in\mathcal{C}\cap\mathcal{T}_{\m{Y}} \;\text{s.t.}\whatmn{T}\neq \m{T}\})\\\nonumber
&=& \expectp{\mathcal{C}}{\mathrm{Pr}_{\m{X}_0,\Delta,\m{Z}}(\{\exists\; \whatmn{T} \in\mathcal{C}\cap\mathcal{T}_{\m{Y}} \;\text{s.t.}\whatmn{T}\neq \m{T}\}|\mathcal{C})}\\\nonumber
&=&\hspace{-0.1in} \expectp{\mathcal{C}}{\sum_{i=1}^{2^{mR_T}}\mathrm{Pr}_{\m{X}_0,\Delta,\m{Z}}(\{\exists\; \whatmn{T} \in\mathcal{C}\cap\mathcal{T}_{\m{Y}} \;\text{s.t.}\whatmn{T}\neq \m{t}(i)\}|\m{t}(i)\;\text{chosen},\mathcal{C})\mathrm{Pr}_{\m{X}_0,\Delta,\m{Z}}(\m{t}(i)\;\text{chosen}|\mathcal{C})} \\
 && + \expectp{\mathcal{C}}{\mathrm{Pr}_{\m{X}_0,\Delta,\m{Z}}(\{\exists\; \whatmn{T} \in\mathcal{C}\cap\mathcal{T}_{\m{Y}}\} |\m{0}\;\text{chosen},\mathcal{C})\mathrm{Pr}_{\m{X}_0,\Delta,\m{Z}}(\m{0}\;\text{chosen}|\mathcal{C})}\nonumber\\
 &=& P_{e,1}  + P_{e,2},
 \label{eq:sumoftwo}
\end{eqnarray}
where $P_{e,1}$ and $P_{e,2}$ represent the first and the second term, respectively. $P_{e,2}$ converges to zero as shown below:
\begin{eqnarray*}
P_{e,2}&=&\expectp{\mathcal{C}}{\mathrm{Pr}_{\m{X}_0,\Delta,\m{Z}}(\{\exists\; \whatmn{T} \in\mathcal{C}\cap\mathcal{T}_{\m{Y}}\} |\m{0}\;\text{chosen},\mathcal{C})\mathrm{Pr}_{\m{X}_0,\Delta,\m{Z}}(\m{0}\;\text{chosen}|\mathcal{C})}\\
&\leq & \expectp{\mathcal{C}}{\mathrm{Pr}_{\m{X}_0,\Delta,\m{Z}}(\m{0}\;\text{chosen}|\mathcal{C})}=\Pr(\mathcal{E}_1) \to 0.
\end{eqnarray*}
Thus we focus on $P_{e,1}$ in~\eqref{eq:sumoftwo}:
\begin{eqnarray}
\nonumber
&&P_{e,1}\\
\nonumber&=&\hspace{-0.15in}\expectp{\mathcal{C}}{\sum_{i=1}^{2^{mR_T}}\mathrm{Pr}_{\m{X}_0,\Delta,\m{Z}}(\{\exists\; \whatmn{T} \in\mathcal{C}\cap\mathcal{T}_{\m{Y}} \;\text{s.t.}\whatmn{T}\neq \m{t}(i)\}|\m{t}(i)\;\text{chosen},\mathcal{C})\mathrm{Pr}_{\m{X}_0,\Delta,\m{Z}}(\m{t}(i)\;\text{chosen}|\mathcal{C})} \\
\nonumber&\overset{(a)}{\leq}& \expectp{\mathcal{C}}{\sum_{i=1}^{2^{mR_T}}\mathrm{Pr}_{\m{X}_0,\Delta,\m{Z}}(\{\exists\; \whatmn{T} \in\mathcal{C}\cap\mathcal{T}_{\m{Y}} \;\text{s.t.}\whatmn{T}\neq \m{t}(i)\}|\m{t}(i)\;\text{chosen},\mathcal{C})2^{-m(R_T-3\epsilon-\epsilon_2)}}\\\nonumber
&=& 2^{-m(R_T-3\epsilon-\epsilon_2)}\expectp{\mathcal{C}}{\sum_{i=1}^{2^{mR_T}}\mathrm{Pr}_{\m{X}_0,\Delta,\m{Z}}(\{\exists\; \whatmn{T} \in\mathcal{C}\cap\mathcal{T}_{\m{Y}} \;\text{s.t.}\whatmn{T}\neq \m{t}(i)\}|\m{t}(i)\;\text{chosen},\mathcal{C})}\\\nonumber
&\leq & 2^{-m(R_T-3\epsilon-\epsilon_2)}\expectp{\mathcal{C}}{\sum_{i=1}^{2^{mR_T}}\mathrm{Pr}_{\m{X}_0,\Delta,\m{Z}}(\m{t}(j)\in\mathcal{T}_{\m{Y}}\;\text{for some}\;j\neq i|\m{t}(i)\;\text{chosen},\mathcal{C})}\\\nonumber
&=& 2^{-m(R_T-3\epsilon-\epsilon_2)}\sum_{i=1}^{2^{mR_T}}\expectp{\mathcal{C}}{\mathrm{Pr}_{\m{X}_0,\Delta,\m{Z}}(\m{t}(j)\in\mathcal{T}_{\m{Y}}\;\text{for some}\;j\neq i|\m{t}(i)\;\text{chosen},\mathcal{C})}\\\nonumber
&=& 2^{-m(R_T-3\epsilon-\epsilon_2)}\sum_{i=1}^{2^{mR_T}}{\mathrm{Pr}_{\m{X}_0,\mathcal{C},\Delta,\m{Z}}(\m{T}(j)\in\mathcal{T}_{\m{Y}}\;\text{for some}\;j\neq i|\text{$i$-th codeword}\;\text{chosen})}\\\nonumber
&=& 2^{-m(R_T-3\epsilon-\epsilon_2)}2^{mR_T}{\mathrm{Pr}_{\m{X}_0,\mathcal{C},\Delta,\m{Z}}(\m{T}(j)\in\mathcal{T}_{\m{Y}}\;\text{for some}\;j\neq 1|\text{$1$st codeword}\;\text{chosen})}\\
&=& 2^{m(3\epsilon+\epsilon_2)}{\mathrm{Pr}_{\m{X}_0,\mathcal{C},\Delta,\m{Z}}(\m{T}(j)\in\mathcal{T}_{\m{Y}}\;\text{for some}\;j\neq 1|\text{$1$st codeword}\;\text{chosen})},
\label{eq:pebd}
\end{eqnarray}
where $(a)$ follows from Lemma~\ref{lem:tdist}.%the same arguments as before: that because we are encoding by joint-typicality, the probability of any codeword is no larger than $2^{-m(I(X_0,T)-\epsilon)}=2^{-m(R_T-3\epsilon-\epsilon_2)}$. 

Now,
\begin{eqnarray}
\nonumber&&{\mathrm{Pr}_{\m{X}_0,\mathcal{C},\Delta,\m{Z}}(\m{T}(j)\in\mathcal{T}_{\m{Y}}\;\text{for some}\;j\neq 1|\text{$1$st codeword}\;\text{chosen})} \\\nonumber
&=& \expectp{\m{X}_0,\m{T}(1)}{\mathrm{Pr}_{\mathcal{C},\Delta,\m{Z}}(\m{T}(j)\in\mathcal{T}_{\m{Y}}\;\text{for some}\;j\neq 1|\m{t}(1)\;\text{chosen},\m{x}_0,\m{t}(1))}\\\nonumber
%&=& \expectp{\m{X}_0,\m{T}(1)}{\mathrm{Pr}_{\mathcal{C},\Delta,\m{Z}}(\m{T}(j)\in \mathcal{T}_{\m{Y}}\text{for some}\;j\neq 1|\m{t}(1)\;\text{chosen},\m{x}_0,\m{t}(1))}\\\nonumber
&=& \expectp{\m{X}_0,\m{T}(1)}{\mathrm{Pr}_{\mathcal{C},\Delta,\m{Z}}\left(\bigcup_{j\neq 1}{\m{T}(j)}\in \mathcal{T}_{\m{Y}}|\m{t}(1)\;\text{chosen},\m{x}_0,\m{t}(1)\right)}\\\nonumber
&\leq& \expectp{\m{X}_0,\m{T}(1)}{\sum_{j\neq 1}\mathrm{Pr}_{\mathcal{C},\Delta,\m{Z}}\left({\m{T}(j)}\in \mathcal{T}_{\m{Y}}|\m{t}(1)\;\text{chosen},\m{x}_0,\m{t}(1)\right)}\\
&\overset{(c)}{\leq}& 2^{mR_T}\expectp{\m{X}_0,\m{T}(1)}{\mathrm{Pr}_{\mathcal{C},\Delta,\m{Z}}\left({\m{T}(2)}\in \mathcal{T}_{\m{Y}}|\m{t}(1)\;\text{chosen},\m{x}_0,\m{t}(1)\right)}.
\label{eq:inter1}
\end{eqnarray}
$(c)$ holds because, by symmetry, $\m{T}(j)$ are conditionally-identically distributed for $j\neq 1$, and the number of such $\m{T}(j)$'s is $2^{mR_T}-1$. %Note that $\mathcal{T}_{\m{y}}$ is simply a set of vectors $\m{t}$ that are jointly typical with the vector $\m{y}$.

In order to calculate the probability above, we obtain an upper bound on the conditional pdf of $\m{T}(2)$,
\begin{eqnarray}
&&\nonumber f\left(\m{T}(2)=\m{t}|\m{t}(1)\;\text{chosen},\m{x}_0,\m{t}(1)\right)\\\nonumber
&\overset{(d)}=&f\left(\m{T}(2)=\m{t}|\m{t}(1)\;\text{chosen},\m{x}_0,\m{t}(1),\m{t}(1)\in\mathcal{T}_{\m{x}_0}\cap\mathcal{B}(M)\right)\\\nonumber
&\overset{(d_1)}=&\frac{\Pr\left(\m{t}(1)\;\text{chosen}|\m{x}_0,\m{t}(1),\m{t}(1)\in\mathcal{T}_{\m{x}_0}\cap\mathcal{B}(M),\m{T}(2)=\m{t}\right)}{\Pr(\m{t}(1)\;\text{chosen}|\m{x}_0,\m{t}(1),\m{t}(1)\in\mathcal{T}_{\m{x}_0}\cap\mathcal{B}(M))}\\
\nonumber &&\times f\left(\m{T}(2)=\m{t}|\m{x}_0,\m{t}(1),\m{t}(1)\in\mathcal{T}_{\m{x}_0}\cap\mathcal{B}(M)\right)\\\nonumber
&\overset{(e)}{=}&\frac{\Pr\left(\m{t}(1)\;\text{chosen}|\m{x}_0,\m{t}(1),\m{t}(1)\in\mathcal{T}_{\m{x}_0}\cap\mathcal{B}(M),\m{T}(2)=\m{t}\right)}{\Pr(\m{t}(1)\;\text{chosen}|\m{x}_0,\m{t}(1),\m{t}(1)\in\mathcal{T}_{\m{x}_0}\cap\mathcal{B}(M))}f\left(\m{T}(2)=\m{t}\right)\\\nonumber
&\leq&\frac{1}{\Pr(\m{t}(1)\;\text{chosen}|\m{x}_0,\m{t}(1),\m{t}(1)\in\mathcal{T}_{\m{x}_0}\cap\mathcal{B}(M))}f\left(\m{T}(2)=\m{t}\right)\\
&\overset{(f)}{\leq}& \frac{1}{2^{-m(3\epsilon + \epsilon_2) - 3}}f(\m{T}(2)=\m{t})=2^{m(3\epsilon + \epsilon_2) + 3}f(\m{T}(2)=\m{t}).
\label{eq:inter2}
\end{eqnarray}
$(d)$ holds because in order for the first codeword $\m{T}(1)$ to be chosen at the encoder, it must be jointly typical with $\m{x}_0$. $(d_1)$ follows from Bayes's rule. $(e)$ holds because $\m{T}(2)$ is drawn independently of $\m{T}(1)$ and $\m{X}_0$. Notice that we are here interested in the pdf of the codeword $\m{T}(2)$, and \textit{not} the probability of it being \textit{chosen} by the encoder for transmission. The argument for $(f)$ is slightly complicated: we first calculate for a given typical $\m{x}_0$, how many $\m{T}$-codewords on average are jointly-typical with it and lie in the correct message bin. From Lemma~\ref{lem:type}, a randomly generated $\m{T}$ is jointly-typical with $\m{X}_0$ with probability at most $2^{-m(I(X_0;T)-3\epsilon)}$ and lies in $M$-th message-bin with probability $2^{-mR}$. Thus, $\expect{|\mathcal{C}\cap\mathcal{T}_{\m{x}_0}\cap\mathcal{B}(M)||\m{x}_0,\m{t}(1),\m{t}(1)\in\mathcal{T}_{\m{x}_0}\cap\mathcal{B}(M)}$, is at most 
\begin{eqnarray}
\nonumber\expect{|\mathcal{C}\cap\mathcal{T}_{\m{x}_0}\cap\mathcal{B}(M)||\m{x}_0,\m{t}(1),\m{t}(1)\in\mathcal{T}_{\m{x}_0}\cap\mathcal{B}(M)}\leq 2^{-m(I(X_0;T)-3\epsilon)}\times 2^{-mR}\times (2^{mR_T}-1) + 1\\
\leq 2^{-m(I(X_0;T)+R-3\epsilon)}\times 2^{m(I(X_0;T)+R+\epsilon_2)}+1 = 2^{m(3\epsilon + \epsilon_2)}+1\leq 2^{m(3\epsilon + \epsilon_2) + 1},
\end{eqnarray}
where the additional term `$+1$' in the first inequality accounts for the additional codeword $\m{t}(1)$ that also lies in $\mathcal{C}\cap\mathcal{T}_{\m{x}_0}\cap\mathcal{B}(M)$; and the last inequality holds for $m$ large. 

Using Markov's inequality, the conditional probability
\begin{eqnarray}
\nonumber
\Pr\bigg(|\mathcal{C}\cap\mathcal{T}_{\m{x}_0}\cap\mathcal{B}(M)| \geq 2^{m(3\epsilon + \epsilon_2)+2}|\m{x}_0,\m{t}(1),\m{t}(1)\in\mathcal{T}_{\m{x}_0}\cap\mathcal{B}(M)\bigg)\\
\leq \frac{\expect{|\mathcal{C}\cap\mathcal{T}_{\m{x}_0}\cap\mathcal{B}(M)||\m{x}_0,\m{t}(1),\m{t}(1)\in\mathcal{T}_{\m{x}_0}\cap\mathcal{B}(M)}}{2^{m(3\epsilon + \epsilon_2)+2}}\leq \frac{1}{2}.
\end{eqnarray}
Thus, with conditional probability at least $\frac{1}{2}$, $|\mathcal{C}\cap\mathcal{T}_{\m{x}_0}\cap\mathcal{B}(M)| < 2^{m(3\epsilon + \epsilon_2)+2}$. Thus 
\begin{eqnarray*}
&&\Pr(\m{t}(1)\;\text{chosen}|\m{x}_0,\m{t}(1),\m{t}(1)\in\mathcal{T}_{\m{x}_0}\cap\mathcal{B}(M))
\\&\overset{(g)}\geq& \Pr(\m{t}(1)\;\text{chosen}|\m{x}_0,\m{t}(1),\m{t}(1)\in\mathcal{T}_{\m{x}_0}\cap\mathcal{B}(M),|\mathcal{C}\cap\mathcal{T}_{\m{x}_0}\cap\mathcal{B}(M)| < 2^{m(3\epsilon + \epsilon_2)+2}) \\&&\times\Pr(|\mathcal{C}\cap\mathcal{T}_{\m{x}_0}\cap\mathcal{B}(M)| < 2^{m(3\epsilon + \epsilon_2)+2}| \m{x}_0,\m{t}(1),\m{t}(1)\in\mathcal{T}_{\m{x}_0}\cap\mathcal{B}(M)) \\
&\geq& \Pr(\m{t}(1)\;\text{chosen}|\m{x}_0,\m{t}(1),\m{t}(1)\in\mathcal{T}_{\m{x}_0}\cap\mathcal{B}(M),|\mathcal{C}\cap\mathcal{T}_{\m{x}_0}\cap\mathcal{B}(M)| < 2^{m(3\epsilon + \epsilon_2)+2}) \times \frac{1}{2}\\
&\geq &\frac{1}{2^{m(3\epsilon + \epsilon_2)+2}} \times \frac{1}{2} = \frac{1}{2^{m(3\epsilon + \epsilon_2)+3}},
\end{eqnarray*}
where $(g)$ follows from $\Pr(A)\geq \Pr(A|B)\Pr(B)$. This proves $(f)$ in~\eqref{eq:inter2}.

Now, from~\eqref{eq:sumoftwo},~\eqref{eq:pebd} and~\eqref{eq:inter1}
\begin{eqnarray*}
&&P_{e,1}\leq 2^{m(3\epsilon+\epsilon_2)}\times 2^{mR_T}\times \expectp{\m{X}_0,\m{T}(1)}{\mathrm{Pr}_{\mathcal{C},\Delta,\m{Z}}\left({\m{T}(2)}\in \mathcal{T}_{\m{Y}}|\m{t}(1)\;\text{chosen},\m{x}_0,\m{t}(1)\right)}\\
&\overset{(g_1)}{=}& 2^{m(3\epsilon+\epsilon_2)}\times 2^{mR_T}\times \\
&&\expectp{\m{X}_0,\m{T}(1)}{\int_{\m{z}}\mathrm{Pr}_{\mathcal{C},\Delta}\left({\m{T}(2)}\in \mathcal{T}_{\m{t}(1)+(1-\alpha) \m{x}_0+\m{z}}|\m{t}(1)\;\text{chosen},\m{x}_0,\m{t}(1)\right)dF(\m{z})}\\
%&\overset{(g)}{=}& 2^{m(3\epsilon+\epsilon_2)}\times 2^{mR_T}\times \expectp{\m{X}_0,\m{T}(1),\m{Y}}{\mathrm{Pr}_{\mathcal{C},\Delta}\left({\m{T}(2)}\in \mathcal{T}_{\m{y}}|\m{t}(1)\;\text{chosen},\m{x}_0,\m{t}(1)\right)}\\
&\overset{(h)}{\leq} &2^{m(3\epsilon+\epsilon_2)}\times 2^{mR_T}\times \expectp{\m{X}_0,\m{T}(1)}{\int_{\m{z}}2^{m(3\epsilon + \epsilon_2)+3}\mathrm{Pr}_{\mathcal{C},\Delta}\left(\m{T}(2)\in\mathcal{T}_{\m{t}(1)+(1-\alpha) \m{x}_0+\m{z}}\right)dF(\m{z})}\\
&{\leq} &2^{m(3\epsilon+\epsilon_2)}2^{m(3\epsilon + \epsilon_2)+3}\times 2^{mR_T}\times \expectp{\m{X}_0,\m{T}(1)}{\int_{\m{z}}\mathrm{Pr}_{\mathcal{C},\Delta}\left(\m{T}(2)\in\mathcal{T}_{\m{t}(1)+(1-\alpha) \m{x}_0+\m{z}}\right)dF(\m{z})}\\
%&\overset{(i)}=&2^{m(3\epsilon+\epsilon_2)}(2^{m(3\epsilon+\epsilon_2)+2}+4)\times 2^{mR_T}\times \expectp{\m{Y}}{\mathrm{Pr}_{\m{T}(2)}\left(\m{T}(2)\in\mathcal{T}_{\m{y}}\right)}\\
&\overset{(i)}\leq&2^{m(6\epsilon+2\epsilon_2)+3}\times 2^{mR_T}\times \expectp{\m{X}_0,\m{T}(1)}{\int_{\m{z}}2^{-m(I(T;Y)-3\epsilon)}dF(\m{z})}\\
&\leq&2^{m(6\epsilon+2\epsilon_2)+3}\times 2^{mR_T}\times 2^{-m(I(T;Y)-3\epsilon)}\\
&=&2^{m(6\epsilon+2\epsilon_2)+3}\times 2^{mR_T}\times 2^{-m(R_T+\epsilon_1-3\epsilon)}=2^{m(6\epsilon+2\epsilon_2)+3}\times 2^{-m(\epsilon_1-3\epsilon)}\\
&=&8\times 2^{-m(\epsilon_1-9\epsilon -2\epsilon_2 )}= 8\times 2^{-m\epsilon}\overset{m\to\infty}\to 0,
%&= &2^{m(2\epsilon+\epsilon_2)+mR_T+m(3\epsilon + \epsilon_2)+1}\times\expectp{\m{X}_0}{\mathrm{Pr}_{\mathcal{C},\Delta,\m{Z}}\left(\wtildemn{T}\in\mathcal{T}_{\m{y}}|\text{1st codeword chosen},\m{x}_0\right)}\\
%&\leq &2^{m(2\epsilon+\epsilon_2)+mR_T+m(3\epsilon + \epsilon_2)+1}\times\expectp{\m{X}_0}{2^{-m((I(T;Y)-\epsilon)}}\\
%&\leq &2^{m(3\epsilon+R_T+2\epsilon_2)+1}\times2^{-m(R_T+\epsilon_1-\epsilon)}=2\times 2^{-m(\epsilon_1-4\epsilon-2\epsilon_2) }\overset{m\to\infty}{\to} 0,
\end{eqnarray*}
where $(g_1)$ follows from the observation that $\m{T}(2)$ is independent independent of $\m{Y}$ conditioned on $\{\m{t}(1)\;\text{chosen},\m{x}_0,\m{t}_1\}$ (because $\m{Y}=\m{t}(1)+(1-\alpha)\m{x}_0+\m{Z}$ depends only on $\m{Z}$ for given $\m{t}(1)$ and $\m{x}_0$), and finally, observing that $\Pr(X\in A) = \int_a \Pr(X\in a|A=a) dF(a)= \int_a\Pr(X\in a) dF(a)$ for indepedent $X$ and $A$. Inequality $(h)$ holds from~\eqref{eq:inter2},  and $(i)$ follows from Lemma~\ref{lem:type}.

Thus the $\m{t}$-codeword can be recovered reliably as long as $R+I(X_0;T)<I(T;Y)$, proving the theorem. 

%%%%%%%%%%%%%%%%%%%%%%%%%%%%%%%
\section{Analytical expressions for upper bounds on costs}
\label{app:upperbound}
\textbf{The $DPC(1)$ strategy ($\beta=0,\alpha=1$)}:\\
%In~\cite{CostaDirtyPaper}, the DPC parameter is chosen to be $\alpha=\alpha_{mmse}=\frac{P}{P+1}$ to achieve the maximum rate of $\frac{1}{2}\lo{1+P}$. Since the cost here is not merely that of input power (an additional $\mathrm{MMSE}$ term is also present), same $\alpha$ may no longer be the optimizing one. For general $\alpha$, as shown in~\cite{CostaDirtyPaper}, the achievable rate is
%\begin{equation}
%R(\alpha) = \frac{1}{2}\lo{\frac{P(P+\sigma^2+1)}{P\sigma^2(1-\alpha)^2+P+\alpha^2\sigma^2}}.
%\end{equation}
With $\alpha=1$ the decoder decodes $\m{X}_1$ perfectly (in the limit $m\rightarrow\infty$), thereby attaining asymptotically zero $\mathrm{MMSE}$. The attained rate (with $\beta=0$) is
\begin{equation}
R(1) = \frac{1}{2}\lo{\frac{P(P+\sigma^2+1)}{P+\sigma^2}}= \frac{1}{2}\lo{P\left(1+\frac{1}{P+\sigma^2}\right)}\geq   \frac{1}{2}\lo{P}.
\end{equation}
Thus, to attain a rate $R$, the cost is upper bounded by the cost attained by $DPC(1)$ strategy, yielding
\begin{equation}
\label{eq:dpcupper}
\avgcost{}_{DPC(1)}:= k^2 2^{2R}+ 0 = k^2 2^{2R}.
\end{equation}
\textbf{The $DPC(\alpha_{Costa})$ strategy ($\beta=0,\alpha=\alpha_{Costa}=\frac{P}{P+1}$)}\\
For this choice of $\alpha$, the achievable rate is well known to equal the channel capacity for an interference-free version of the channel~\cite{CostaDirtyPaper}
\begin{equation}
R(\alpha_{Costa}) =  \frac{1}{2}\lo{1+P},
\end{equation}
and the required power is, therefore, $P=2^{2R}-1$. The expression for MMSE in the estimation of $\m{X}_1=\m{X}_0+\m{V}$ can be unwieldy because it is estimated using $\m{T}=\m{V}+\alpha\m{X}_0$ as well as $\m{Y}=\m{X}_1+\m{Z}$. For analytical simplicity, we use two upper bounds.  Instead of estimating $\m{X}_1$ using \textit{both} $\m{Y}$ and $\m{T}$, we use just $\m{Y}$, or just $\m{T}$. When using just $\m{Y}$, note that the variance of $\m{X}_1=\m{X}_0+\m{V}$ is $m(\sigma^2+P)$ because of the asymptotic orthogonality of $\m{X}_0$ and $\m{V}$. Thus the MMSE error is $\frac{\sigma^2+P}{\sigma^2+P+1} < 1$. 

When estimating $\m{X}_1$ from $\m{T}$, assuming asymptotically perfect decoding of $\m{T}$, the MMSE error is 
\begin{eqnarray}
\label{eq:secondlowerbound}
\mathrm{MMSE} &=& \expect{X_1^2}-\frac{\left(\expect{X_1T}\right)^2}{\expect{T^2}}=P+\sigma^2 - \frac{(P+\alpha\sigma^2)^2}{P+\alpha^2\sigma^2}\nonumber\\
& = & \textcolor{red}{\frac{P^2+\alpha^2 P \sigma^2 + P\sigma^2+ \alpha^2\sigma^4   - P^2  -\alpha^2\sigma^4 -2\alpha P \sigma^2   }{P+\alpha^2\sigma^2}}\nonumber\\
&=& \textcolor{red}{\frac{\alpha^2P\sigma^2 + P\sigma^2 - 2 \alpha P \sigma^2}{P+\alpha^2\sigma^2}}\nonumber\\
& = & \frac{P\sigma^2  (1-\alpha)^2}{P+\alpha^2\sigma^2}\nonumber\overset{\left(\alpha=\frac{P}{P+1}\right)}{=}  \frac{ P\sigma^2\left(\frac{1}{P+1}\right)^2   }{P + \frac{P^2}{(P+1)^2}\sigma^2}=\frac{\sigma^2}{(P+1)^2+P\sigma^2}\nonumber\\
& \overset{\left(P=2^{2R}-1\right)}{=} & \frac{\sigma^2}{2^{4R}+(2^{2R}-1)\sigma^2}.
\end{eqnarray}
Thus, the cost for $DPC(\alpha_{Costa})$ strategy is upper bounded by
\begin{equation}
\avgcost_{DPC(\alpha_{Costa})} := k^2 (2^{2R} - 1) + \min \left\{ 1, \frac{\sigma^2}{2^{4R}+(2^{2R}-1)\sigma^2}\right\}
\end{equation}
\textbf{Host signal cancelation ($\beta=1$)}:\\
In this strategy, we first force the host signal to zero (using $\beta=1$), and then add a codeword to communicate across the channel. The resulting strategy is thus equivalent to trivial dirty-paper coding where the interference is zero. For consistency, we continue to denote the remaining part of the input by $\m{V}_{\mathrm{dpc}}$. Since the message (and hence the input $\m{V}_{\mathrm{dpc}}$) is independent of the initial state $\m{X}_0$, the total power required is the sum of the powers of the codeword and the initial state. Because the channel is now just an AWGN channel, the required $\m{V}_{\mathrm{dpc}}$ power to communicate at rate $R$ is $2^{2R}-1$. The required total power is $P=\sigma^2+2^{2R}-1$, the asymptotic average reconstruction error is zero, and the required cost is bounded by
\begin{equation}
\avgcost{}_{cancel}\leq k^2 (\sigma^2+2^{2R}-1).
\end{equation}

\section{Approximate optimality of the proposed strategies}
\label{app:approx}
We divide the parameter space into four regions. Let $\kappa:=\frac{\sigma^2 2^{2R}}{(\sigma+\sqrt{P})^2+1}$. From~\eqref{eq:looselower}, a lower bound on the total asymptotic costs is given by
\begin{equation}
\label{eq:loosened}
\avgcost_{\text{opt}}\geq \inf_{P\geq 2^{2R}-1}k^2 P + \left(\left(   \sqrt{\kappa} - \sqrt{P}   \right)^+ \right)^2,
\end{equation}
where the lower limit on $P$ follows from the observation that to communicate reliably at rate $R$, it requires a power of at least $2^{2R}-1$~\cite{CostaDirtyPaper}. In the proof below, define $P^*$ to be the optimizing value of $P$ in~\eqref{eq:loosened}.

\textit{Case 1}: $R\geq\frac{1}{4}$ (and any $P^*,\sigma^2$).\\
We use the $DPC(1)$ upper bound from~\eqref{eq:dpcupper} of $k^2 2^{2R}$. The lower bound in~\eqref{eq:loosened} larger than $k^2(2^{2R}-1)$. The ratio of upper and lower bounds is therefore smaller than
\begin{eqnarray*}
\frac{k^2 2^{2R}}{k^2 (2^{2R}-1)}=\frac{2^{2R}}{2^{2R}-1}=1+ \frac{1}{2^{2R}-1}\overset{\left(R\geq \frac{1}{4}\right)}{\leq} 1 + \frac{1}{\sqrt{2}-1}= \frac{\sqrt{2}}{\sqrt{2}-1}\approx 3.4<4.
\end{eqnarray*}

\textit{Case 2}: $P^*\geq \frac{2^{2R}}{16}$ (and any $\sigma^2$).\\
Again, we use the $DPC(1)$ upper bound of $k^22^{2R}$. The lower bound is larger than $k^2P^*\geq k^2  \frac{2^{2R}}{16}$. Thus, the ratio of upper and lower bounds for this case is smaller than $16$.

\textit{Case 3}: $R<\frac{1}{4}$, $P^*< \frac{2^{2R}}{16}, \sigma^2 >1$.\\
For the lower bound, note that
\begin{eqnarray*}
\kappa &=& \frac{\sigma^2 2^{2R}}{(\sigma+\sqrt{P^*})^2+1}\overset{(a)}{\geq}  \frac{2^{2R}}{(1+\sqrt{P^*})^2+1}\\
&\overset{P^*< \frac{2^{2R}}{16}}{\geq} & \frac{2^{2R}}{\left(1+\frac{2^R}{4}\right)^2+1}
\overset{0<R<\frac{1}{4}}{\geq} \frac{1}{\left(1+\frac{2^{\frac{1}{4}}}{4}\right)^2 + 1}\approx 0.3727> 0.37,
\end{eqnarray*}
where $(a)$ follows from the fact that $\sigma^2\geq 1$ and that the expression on left-hand-side of $(a)$ is increasing in $\sigma^2$ (to see this, divide the numerator and the denominator by $\sigma^2$).

Thus, 
\begin{eqnarray*}
\avgcost{}_{\text{opt}}&\geq& k^2 (2^{2R}-1)+ \left(\left(\sqrt{\kappa}-\sqrt{P^*}\right)^+\right)^2\\
&\geq &k^2 (2^{2R}-1)+ \left(\sqrt{0.37}-\frac{2^R}{4}\right)^2 \overset{\left(R<\frac{1}{4}\right)}{\geq}k^2 (2^{2R}-1)+ \left(\sqrt{0.37}-\frac{2^\frac{1}{4}}{4}\right)^2\\
&\approx&k^2 (2^{2R}-1)+ 0.0967 >k^2 (2^{2R}-1)+0.09.
\end{eqnarray*}
Upper bound of $DPC(\alpha_{Costa})$ is smaller than $k^2(2^{2R}-1)+1$. Thus the ratio in this case is smaller than $\frac{1}{0.09}\approx 11.11<12$.

\textit{Case 4}: $R<\frac{1}{4}$, $P^*< \frac{2^{2R}}{16}, \sigma^2 \leq1$.\\ 
\textit{Case 4a}: If $P^*\geq\frac{\sigma^2}{8}$.

Observe that $P^*\geq 2^{2R}-1$ and $P^*\geq \frac{\sigma^2}{8}$. Thus, $P^*>\frac{2^{2R}-1 +\frac{\sigma^2}{8}}{2}=\frac{2^{2R}-1}{2}+\frac{\sigma^2}{16}$. Thus, the lower bound
\begin{equation}
\avgcost_{\text{opt}}\geq k^2P^*\geq  k^2\left(\frac{2^{2R}-1}{2}+\frac{\sigma^2}{16}\right).
\end{equation}
The upper bound $\avgcost_{cancel}\leq k^2 \left(2^{2R}-1 + \sigma^2\right)$. Thus the ratio of upper and lower bounds is smaller than $16$. \\
\text{Case 4b}: If $P^*<\frac{\sigma^2}{8}$.

%If $\sigma^2 < 20(2^{2R}-1)$, using the upper bound of host signal cancelation $J_{cancel}$, the cost is smaller than
%\begin{eqnarray*}
%\cost{}_{\min} \leq k^2(\sigma^2 + 2^{2R}-1) \leq 21 k^2 (2^{2R}-1).
%\end{eqnarray*}
%Since the lower bound is larger than $k^2(2^{2R}-1)$, the ratio is smaller than $21$.
%
%\textit{Case 4b}: If $20(2^{2R}-1)<\sigma^2\leq1$, then the straight coding upper bound yields 
%\begin{equation}
%\cost{}_{\min}\leq k^2(\sigma^2 + 2^{2R}-1) \leq \frac{21}{20} k^2 \sigma^2.
%\end{equation}
In this case,
\begin{eqnarray*}
\kappa &=& \frac{\sigma^2 2^{2R}}{\left(\sigma + \sqrt{P^*}\right)^2+1}\overset{R\geq 0}{\geq}  \frac{\sigma^2}{\left(\sigma + \sqrt{P^*}\right)^2+1}\\
&\overset{\sigma\leq 1, P^*\leq \frac{2^{2R}}{16} <\frac{\sqrt{2}}{16} }{\geq} & \frac{\sigma^2}{(1+\frac{2^{\frac{1}{4}}}{4})^2 +1}\approx \frac{\sigma^2}{2.6830}\geq \frac{\sigma^2}{2.69}.
\end{eqnarray*}
Thus the lower bound is larger than
\begin{eqnarray*}
\avgcost{}_{\text{opt}} &\geq &k^2 (2^{2R}-1) + \left(\left(    \sqrt{\frac{\sigma^2}{2.69}}  - \sqrt{P^*}    \right)^+ \right)^2\\
&\overset{(P^*\leq\frac{\sigma^2}{8})}{\geq}  & k^2(2^{2R}-1) + \left(\left(    \sqrt{\frac{\sigma^2}{2.69}}  - \sqrt{\frac{\sigma^2}{8}}   \right)^+ \right)^2\\
& \approx&  k^2(2^{2R}-1) + 0.0656\sigma^2\geq  k^2(2^{2R}-1) + 0.065\sigma^2.
\end{eqnarray*}
The upper bound is based on $DPC(\alpha_{Costa})$. Using~\eqref{eq:secondlowerbound}, an upper bound is 
\begin{eqnarray*}
\avgcost{}_{\text{opt}} &\leq & k^2(2^{2R}-1) + \frac{\sigma^2}{2^{4R}+ (2^{2R}-1)\sigma^2}\\
&\overset{R\geq 0}{\leq} & k^2 (2^{2R}-1) + \sigma^2.
\end{eqnarray*}
The ratio is smaller than $\frac{1}{0.065}\approx 15.39 <16$, and over the entire space, the ratio of asymptotic upper and lower bounds is smaller than $16$.

%%%%%%%%%%%%%%%%%%%%%%%%%%%%%%%%%%%%%%%%%%%%%%%%%%%%%%%%%%
\bibliographystyle{IEEEtran}
\bibliography{IEEEabrv,MyMainBibliography}

\end{document}